\documentclass{amsart}
\usepackage{amsmath, amsthm, amscd, amsfonts, latexsym, amssymb, graphicx, color}

\newcommand{\beqa}{\begin{eqnarray*}}
\newcommand{\eeqa}{\end{eqnarray*}}
\newcommand{\beqn}{\begin{eqnarray}}
\newcommand{\eeqn}{\end{eqnarray}}

\newcommand{\bs}{\boldsymbol}

\newcommand{\C}{\mathbb C}

\newcommand{\R}{\mathbb R}

\newcommand{\mcH}{\mathcal H}

\newcommand{\mcP}{\mathcal P}

\newcommand{\mcZ}{\mathcal Z}

\newcommand{\tf}{\tfrac}

\newcommand{\g}{\gamma}

\newcommand{\lb}{\label}
\newcommand{\rf}{\ref}
\newcounter{cnt1}
\newcounter{cnt2}
\newcounter{cnt3}
\newcommand{\blr}{\begin{list}{$($\roman{cnt1}$)$}
 {\usecounter{cnt1} \setlength{\topsep}{0pt}
 \setlength{\itemsep}{0pt}}}
\newcommand{\bla}{\begin{list}{$($\alph{cnt2}$)$}
 {\usecounter{cnt2} \setlength{\topsep}{0pt}
 \setlength{\itemsep}{0pt}}}
\newcommand{\bln}{\begin{list}{$($\arabic{cnt3}$)$}
 {\usecounter{cnt3} \setlength{\topsep}{0pt}
 \setlength{\itemsep}{0pt}}}
\newcommand{\el}{\end{list}}
\newtheorem{thm}{Theorem}[section]

\newtheorem{cor}[thm]{Corollary}
\newtheorem{ex}[thm]{Example}

\newtheorem{Def}[thm]{Definition}

\newtheorem{rem}[thm]{Remark}
\newcommand{\Rem}{\begin{rem} \rm}
\newcommand{\bdfn}{\begin{Def} \rm}
\newcommand{\edfn}{\end{Def}}
\newcommand{\tx}{\text}

\newcommand{\ba}{\begin{array}}
\newcommand{\ea}{\end{array}}

\numberwithin{equation}{section}
\date{}
\begin{document}
\title{\bf The Einstein Dual Theory of Relativity}
\author{Tepper L. Gill}
\address[Tepper L. Gill]{Department of EECS, Mathematics and Computational Physics Laboratory, Howard University,
Washington DC 20059 USA,  {\tt tgill@howard.edu; tgill@access4less.net}}
\author[Ares de Parga]{G. Ares de Parga}
\address[Gonzalo Ares de Parga]{Departmento de F\'{\i}sica, Escuela Superior de F\'{\i}sica y
Matem\'{a}ticas, Instituto Polit\'{e}cnico Nacional COFAA\\
Edif 9, U. P. Adolfo L\'{\o}pez Mateos, Zacatenco, Lindavista, 07738,
M\'{e}xico D.F., M{\'e}xico\\ \tt{E-mail~:} \tt{gadpau@hotmail.com}}
%\abstract{ write abstract here}
\begin{abstract}
{This paper is a comparison of the Minkowski, Einstein and Einstein dual theories of relativity.  The dual is based on an identity relating the observer time and the proper time as a contact transformation on configuration space, which leaves phase space  invariant. The theory is dual in that, for a system of $n$ particles, any inertial observer has two unique sets of global variables $({\bf{X}}, t)$ and $({\bf{X}}, \tau)$ to describe the dynamics.  Where ${\bf{X}}$ is the (unique) canonical center of mass.  
In the $({\bf{X}}, t)$ variables, time is relative and the speed of light is unique, while in the $({\bf{X}}, \tau)$ variables, time is unique and the speed of light is relative with no upper bound. The two sets of particle and Maxwell field equations are mathematically equivalent, but the particle wave equations are not. The dual version contains an additional longitudinal radiation term that appears instantaneously with  acceleration,  does not depend on the nature of the force and the Wheeler-Feynman absorption hypothesis is a corollary.
 
The homogenous and isotropic nature of the universe is sufficient to prove that  a unique definition of Newtonian time exists  with zero set at the big bang. The isotopic dual of $\R$ is used to improve the big bang model, by providing an explanation for the lack of antimatter in our universe, a natural arrow for time, conservation of energy, momentum and angular momentum.  This also solves the flatness and horizon problems without inflation.

We predict that radiation from a betatron (of any frequency) will not produce photoelectrons,  that matter and antimatter are gravitationally repulsive and that data from distant sources does not have a unique physical interpretation.  We provide a table showing the differences between the Minkowski, Einstein and dual versions of the special theory.}
 
\end{abstract} 
\maketitle
\section{Background and History}
In the beginning of the last century, the problem of reconciling the transformation properties of the Newtonian and Maxwell theories was the great concern.  We are now starting a new century and this problem is still with us, along with a host of new ones. 

Einstein, Lorentz and Poincar{\' e} all faced these problems directly.   In the course of his investigation, Lorentz \cite{1, 2}  showed that all of the macroscopic phenomena of optics  and electrodynamics can be explained from a detailed analysis of the microscopic behavior of electrons and ions.   Poincar{\' e} discovered an error in Lorentz's analysis and realized that, after correction the transformations formed a group, which he named for Lorentz \cite{3}.  By 1906 Poincar{\'e} had already shown that, if time is treated as an imaginary coordinate, the Lorentz group can be treated as a rotation in four-dimensional space and introduced the metric (proper-time) later introduced by Minkowski (see \cite{4}).  Poincar{\' e}'s strong background in physics and philosophy of science, in addition to his insight and understanding of the difference between mathematics and physics helped him to resist the temptation to use this ``physically unjustified" mathematical observation as a (necessary) tool for the representation of physical reality.

Independently,  Einstein related the photoelectric effect to the quantum ideas of Planck and derived the Lorentz transformations from basic kinematical arguments, as opposed to the symmetry properties of Maxwell's equations (as was done by Lorentz).  Einstein chose this approach because he did not believe Maxwell's theory would survive the existence of photons (see Brown  \cite{5}).  

Observing that the constant $c$ appears in Maxwell's equations for all inertial observers,  Einstein \cite{6} realized that a formal postulate on the velocity of light was necessary.    He proposed that all physical theories should satisfy the (well-known) postulates of special relativity:

\medskip  {
	{(\bf{1})} The physical laws of nature and the results of all experiments are independent
 
\quad of the 	particular inertial frame of the observer (in which the experiment is

\quad performed).
   
\medskip
	{(\bf{2})} The speed of light in empty space is constant and is independent of the

\quad motion of the 	source or receiver.}         

Minkowski was the first to suggest that Poincar{\'e}'s discovery be made a fundamental part of the special theory.   He was a well-known number theorist with few accomplishments in physics and a strong belief in Hilbert's program to geometrize physics \cite{7}.   (A  complete analysis of Minkowski's motivation, his knowledge of Poincar{\'e}'s work and his background in physics can be found in Walters \cite{8}.) {\it{Thus, we make explicit Minkowski's unacknowledged additional postulate to the special}} {\it{theory of relativity}}:

 \medskip
	{(\bf{3})} { The correct implementation of the first two postulates requires that time
be 

\quad treated as a fourth coordinate, and the relationship between components so

\quad constrained as to satisfy the invariance induced by the Lorentz group, using 

\quad the proper time (Minkowski space).}

\subsection{Newtonian Mechanics}
As reported by Sommerfeld, Minkowski knew that the differential of proper time is not an exact one-form (see the notes in \cite{9}).  Thus, he introduced the co-moving observer as a substitute in order to use it as a metric. 

Einstein, Lorentz, Poincar\' e,  Ritz and other important thinkers on the subject maintained their belief that space and time had distinct physical properties.  Einstein was the first to oppose Minkowski's postulate openly.     As noted by Sommerfeld, Einstein was critical of Minkowski's implicit assumption that no physics was lost by constraining the differential of proper time.  Einstein and Laub later published two papers on electrodynamics, which offered a different approach,  was simpler and did not depend on the spacetime formalism (see \cite{10}, \cite{11}).  They argued that the spacetime  formalism was complicated, required additional assumptions and did not add any new physics.

Sommerfeld later simplified Minkowski's complicated formulation, making it easy for physicists to understand the new tensor methods.    The new trend towards abstracting concepts and methods  automatically  made the theory attractive to mathematicians.  This made Minkowski's ideas  even more popular and helped to bring  them to the attention of the masses.  In this air of euphoria, it was not noticed that the theory did not work for two or more particles and thus was far from an  extension of Newton's mechanics. (This is the true cause of the twin paradox.) By the time problems in attempts to merge the special theory with quantum mechanics forced researchers to take a new look at the foundations of electrodynamics,  Minkowski's postulate had become sacred.  When Einstein considered the extension of the special to the general theory, he was only interested  in one which extended Minkowski's postulate (see Pais \cite{12} and  \cite{13}). 
  
Once it was accepted that the proper Newtonian theory should be invariant under the Lorentz group, the problem was ignored until after World War Two when it was realized that quantum theory did not solve the problems left open by the classical theory.   

In classical electrodynamics, Dirac partially by-passed many of the problems by replacing particles with fields (see \cite{14}).   However, this approach led to the first example of a divergent theory (infinite self-energy).  This divergency was the main motivation for the Wheeler-Feynman approach to classical electrodynamics (see \cite{15}).  Their theory solved the divergency problem, but could not be used as the foundations for quantum theory.  However, it still give Feynman a different approach to quantum electrodynamics (QED).

The failure to solve the classical problem forced researchers to use the Dirac theory as the basis for relativistic quantum mechanics and QED.  This approach maintained the infinite self-energy divergence and introduced a few others.  These problems were later by-passed by Feynman, Schwinger and Tomonaga in the late 1940's leading to the great successes of that era.  It was expected that the mathematicians would eventually find the correct theory to justify the methods of QED.  However, by the early 1980's, it became clear that this was not to be and the next generation pinned their hopes on string theory as the best way forward.  At this time, we have no definitive answers.  The development of the electro-weak theory and the standard model have each added additional problems. Thus another serious look at the classical situation can't make things worst.

\subsubsection{The $2$-particle time problem} 
In order to understand the  two particle time problem, we consider two inertial observers $O$ and $O'$. Without loss, assume both clocks begin when their origins coincide and $O'$ is moving with uniform velocity ${\bf v}$ as seen by $O$.  Let two particles, each the source of an electromagnetic field, move with velocities ${\bf w}_i \,  (i=1,2)$, as seen by $O$, and ${\bf w}_i' \,  (i=1,2) $, as seen by $O'$, so that:  
\beqn 
\begin{gathered}
{\bf x}'_i={\bf x}_i-\gamma ({\bf
v}){\bf v}t+(\gamma ({\bf v})-1)({\bf x}_i\cdot {\bf v}/\left\| {\bf v} \right\|^2){\bf v} \hfill \\ 
{\rm{and}}  \hfill \\ 
{\bf x}_i={\bf x}'_i+\gamma ({\bf v}){\bf v}t'+(\gamma ({\bf
v})-1)\left( {{\bf x}'_i\cdot {\bf v}/\left\| {\bf v}\right\|^2} \right){\bf v}, \hfill \\ 
\end{gathered} 
\eeqn
with $\gamma ({\bf v})=1/\left[ {1-\left( {{\bf v}/c}\right)^2} \right]^{1/2}$
represent the spacial Lorentz transformations between the corresponding observers. 
Thus, there is clearly no problem in requiring that the positions transform as
expected.  However, when we try to transform the clocks, we see the problem at once
since we must have, for example, 
\beqn
 t'=\gamma ({\bf v})\left( {t-{\bf x}_1\cdot {\bf v}/c^2} \right)\ \ {\rm{and}}\ \ t'=\gamma
({\bf v})\left({t-{\bf x}_2\cdot {\bf v}/c^2}\right). 
\eeqn   
This is clearly impossible unless ${{\mathbf{x}}_1} \cdot {\mathbf{v}} = {{\mathbf{x}}_2} \cdot {\mathbf{v}}$.  Thus  we cannot use the observer clock to share information (with other observers) about two or more particles. Thus, we conclude that we cannot use the observer's clocks to maintain the first postulate.   
\subsubsection{The $n$-particle position problem}  
Pryce was the first to study the center of mass problem for $n$ particles (see \cite{16}). He concluded that there are three possibilities, but only one is canonical and available to all observers. His representation led to the implication that the canonical center-of-mass cannot be the three-vector part of a four-vector. (This problem is almost seventy five years old.) The analysis of Pryce will be discussed fully in the next section and the problem will be made explicit.

After Pryce's investigation, Bakamjian and Thomas  showed that they could construct a many-particle quantum theory that satisfied Einstein's two postulates, but not Minkowski's (see  \cite{17}).  They further suggested that, with the addition of  Minkowski's postulate, their theory would only be compatible with free particles.  
\subsubsection{No-Interaction}
There are two major no-interaction theorems: the first was due to Haag \cite{18} and applies to the foundations of quantum field theory. Today It is often confused with the one proved by Currie et al  \cite{19}, which shows that Bakamjian-Thomas  were correct. The theorem has since been extended to an arbitrary number of particles by Leutwyler {\cite{20}}. We present the general form. (For a recent version, see \cite{40}.)    
\begin{thm}{\rm{(No-Interaction Theorem)}} Consider a system of particles $\left\{ {\left( {{{\mathbf{x}}_i},{{\mathbf{p}}_i}} \right)} \right\}_{i = 1}^n$  defined  on $\R^{3n} \times \R^{3n} $ (phase space).  Supposed that the following is satisfied:
\begin{enumerate}
\item   The system has a Hamiltonian representation. 
\item  The system has a canonical representation of the Poincare group.
\item  Each ${{{\mathbf{x}}_i}}$ is the vector part of a four-vector.
\end{enumerate}
Then these assumptions are only compatible with free particles.
\end{thm}
All attempts to keep Minkowski's postulate, avoid The No-Interaction Theorem, include Newtonian mechanics and merge with quantum mechanics have failed.

In 1963, in the same paper where he suggested the study of strings, Dirac \cite{21} openly challenge the fundamental nature of Minkowski's postulate.   He  wrote:
\begin{quotation}
What appears to our consciousness is really a three-dimensional section of the four-dimensional picture. We must take a three-dimensional section to give us what appears to our consciousness at one time; at a later time we shall have a different three-dimensional section. The task of the physicist consists largely of relating events in one of these sections to events in another section referring to a later time. Thus the picture with four-dimensional symmetry does not give us the whole situation. This becomes particularly important when one takes into account the developments that have been brought about by quantum theory. Quantum theory has taught us that we have to take the process of observation into account, and observations usually require us to bring in the three-dimensional sections of the four-dimensional picture of the universe. ...

when one looks at gravitational theory from the point of view of the sections, one finds that there are some degrees of freedom that drop out of the theory. The gravitational field is a tensor field with 10 components. One finds that six of the components are adequate for describing everything of physical importance and the other four can be dropped out of the equations. One cannot, however, pick out the six important components from the complete set of 10 in any way that does not destroy the four-dimensional symmetry.
\end{quotation}
 
The invariance requirement for Maxwell's equations can be satisfied by the Lorentz group without Minkowski's   postulate (see \cite{22}).   We conclude that Minkowski's  postulate imposes an additional condition on Einstein's special theory of relativity for one particle, but fails completely for two or more particles at the classical level and creates even more problems at the quantum level. 

\subsection{The 2.7 $^{\circ}$K mbr and Mach's Principle}
Two years after Dirac's paper, Penzias and Wilson discovered, the 2.7 $^{\circ}$K microwave background radiation (mbr).  It has been known since, that this radiation defines a unique preferred frame of rest, which exists throughout the universe and is available to all observers (see \cite{23}).  This radiation  is highly isotropic with anisotropy limits  set at $0.001\%$.    Direct measurements have been made of the velocity of our Solar System and Galaxy relative to the mbr (370 and 600 $km/sec$ respectively, see Peebles \cite{24} ). 

Peebles has  suggested that, the special theory is valid with or without a preferred frame, so that the mbr does not violate the special theory.  However, this statement is not obvious, in addition,  general relativity predicts that at each point one can adjust their acceleration locally to find a freely falling frame where the special theory holds. In this frame, observers with constant velocity are equivalent.  Thus, according to the general theory there is an infinite family of freely falling frames. The Penzias and Wilson findings show that, one can set the acceleration equal to zero.
\subsection{Major Foundational Problems}
There are many opinions about the role of mathematics in physics.  In this section, we first define  the proper role of theoretical physics and the proper role of mathematics in relationship to physics.  We then identify seven major problems, that must be solved if we are to  provide a solid foundation for physics to move forward in the twenty first century. 
\begin{rem}Many may think that the role of theoretical physics and of mathematics in relationship to physics is obvious and strongly question the necessity for this section.  However, we live at a time when the majority view is that the previously unsolved problems are of no real concern, pointing out the great empirical successes of the past.  

These unsolved problems have been with us so long, that the role and view of theoretical physics as a tool for (and a part of) science is in question. The recent book by Frisch \cite{41} on classical electrodynamics not only provides a clear discussion of the problems, and internal (mathematical) inconsistances,  he further assumes they have no solution and suggests that this state of affairs be accepted as a natural part of the theoretical landscape. Similar sentiments have been expressed by Schweber \cite{42} concerning the well-known difficulties in QED.
\end{rem} 
\subsubsection{Theoretical Physics}
The objective of theoretical physics is to design faithful representations or models of the physical world.  These designs must be able to describe the cause effect relationships observed in experiment and, they must be physically and mathematically consistent.   To be useful, these designs must also be constructed using a minimal number of variables and parameters.  
\newline  The basic postulate is that:

\indent \; \; \; {  Mathematics is the correct tool for the design, analysis and certification  
\newline \indent  \quad \quad  of the consistency of representations of physical reality}.  
\subsubsection{Mathematics}
From the (restricted) view of theoretical physics, mathematics is defined as:
\begin{enumerate}
\item A tool for the design of internally consistent languages and structures.
\item A tool for the design of representations of the physical world.
\item A tool for the qualitative and quantitative analysis of data about and  representations of the physical world.
\end{enumerate}
In some cases, mathematical languages and structures, designed for other purposes, have  become perfect tools for certain parts of physics (e.g., group theory, probability theory, statistics).   However, the most useful languages and structures have been those specifically designed for physics (e.g., geometry, calculus, differential and partial differential equations, vector analysis,   geometric algebra and isotopes).  

Thus, the role of mathematics in theoretical physics is that of a tool. This is where there appears to be confusion.  We should be clear that, any mathematical model resulting from a theoretical design is not physical reality, but at most, the best representation we can design at this point in our intellectual evolution. (Anyone seeking absolute understanding or knowledge will not find it in physics.) 
\begin{rem}
From this perspective, ``mathematics is amazingly effective in physics" because it was designed for just that purpose.
\end{rem}
We have identified seven major problems that must be faced directly if we want to design a consistent structure, which will provide a clear path forward in the twenty first century.  Any design:
\begin{enumerate}
\item {{{must be compatible with the two postulates of Einstein};}} 
\item {{{must be compatible with Newtonian mechanics};}}
\item {must be compatible with classical electrodynamics;} 
\item {{must be compatible with the  2.7 $^{\circ}$K MBR};}
\item {{{ must be compatible with quantum mechanics};}}
\item {{{must be compatible with the results of general relativity} and, }}
\item {{{must be mathematically consistent}. }}
\end{enumerate}
These are the seven  foundational pillars of theoretical physics.   In the remaining sections of the paper, we develop the first four requirements above while insuring that they satisfy the last requirement.  The remaining requirements are part of an ongoing effort.
\section{Newton and Einstein without Minkowski} 
We begin with the design of a general model that includes Newton and Einstein.  We assume a classical interacting system of $n$-particles  defined in terms of physical variables and observed by $O$ in an inertial frame.   Observer $O$ is able to identify each particle and attach a vector ${\bf{x}}_i$ to the $i^{\tx{th}}$ particle, denoting its spacial distance to the origin.   
\subsubsection{One-Particle Clock}
First, we construct a unique clock for each particle.   Let $O$ observe the dynamics of particle $i$ using coordinates $({\bf{x}}_i, t)$.  If ${\mathbf{w}}_i$ is the velocity of particle $i$ as seen by $O$, let ${\gamma ^{ - 1}}\left( {\mathbf{w}}_i \right) = \sqrt {1 - {{{{\mathbf{w}}_i^2}} \mathord{\left/
 {\vphantom {{{{\mathbf{w}}_i^2}} {{c^2}}}} \right.
 \kern-\nulldelimiterspace} {{c^2}}}} $.  The $i^{\tx{th}}$ particle proper time is defined by:
\beqn
d\tau_i  ={\g}^{-1}({\bf{w}}_i)dt,\quad {\mathbf{w}}_i = \frac{{d{\mathbf{x}}_i}}{{dt}},\quad d{\tau_i^2} = d{t^2} - \tfrac{1}{{{c^2}}}d{{\mathbf{x}}_i^2}.
\eeqn
We  can also rewrite the last term to get:
\beqn\lb{id0}
\quad \quad d{t^2} = d{\tau_i^2} + \tfrac{1}{{{c^2}}}d{{\mathbf{x}}_i^2}, \Rightarrow cdt = \left( {\sqrt {{{\mathbf{u}}_i^2} + {c^2}} } \right)d\tau_i ,\quad {\mathbf{u}}_i = \frac{{d{\mathbf{x}}_i}}{{d\tau_i }}= {\g}({\bf{w}}_i){{\mathbf{w}}_i}.
\eeqn
If we let $b_i={\sqrt {{{\mathbf{u}}_i^2} + {c^2}} }$, the second term in equation (\rf{id0}) becomes $cdt =b_id\tau_i$.  This leads to our first identity:
\beqn{\lb{id1}}
\frac{1}{c}\frac{d}{{dt}} \equiv \frac{1}{{{b_i}}}\frac{d}{{d{\tau _i}}}
\eeqn
This identity provides the correct way to define the relationship between the proper time and the observer time for the $i^{th}$ particle.   If we apply the identity to ${\bf{x}}_i$, we obtain our second new identity, which shows that the transformation leaves the configuration (or tangent) space invariant:
\beqn{\lb{id1a}}
\frac{{{{\mathbf{w}}_i}}}{c} = \frac{1}{c}\frac{{d{{\mathbf{x}}_i}}}{{dt}} \equiv \frac{1}{{{b_i}}}\frac{{d{{\mathbf{x}}_i}}}{{d\tau_i }} = \frac{{{{\mathbf{u}}_i}}}{{{b_i}}}.
\eeqn
The new particle coordinates are $({\bf{x}}_i, \tau_i)$.  In this representation, the position ${\bf{x}}_i$ is uniquely defined relative to $O$, while $\tau_i$ is uniquely defined by the $i^{\tx{th}}$ particle. Using ${\g}({\bf{w}}_i)=H_i/m_ic^2$, we can also write  $d{\tau _i} = \left( {{{{m_i}{c^2}} \mathord{\left/ {\vphantom {{{m_i}{c^2}} {{H_i}}}} \right.  \kern-\nulldelimiterspace} {{H_i}}}} \right)dt$.   The $i^{\tx{th}}$ particle momentum can be represented as ${\bf{p}}_i=m_i \g({\bf{w}}_i){\bf{w}}_i=m_i{\bf{u}}_i$, where $m_i$ is the particle rest mass.  Thus, the phase space variables are left invariant.
\subsubsection{Many-Particle Clock} 
To construct the many-particle clock, we suppose the interacting particles have proper clocks $\tau_i$, Hamiltonians  $H_i $ and total Hamiltonian $H = \sum_{i = 1}^n {H_i }$.  We define the effective mass $M$ and total momentum $ {\bf P}$ by
\[ 
Mc^2  = \sqrt {H^2  - c^2 {\bf P}^2 },\quad {\bf P} = \sum\limits_{i = 1}^n {{\bf
p}_i }.  
\]          
We can then represent $H$ as $H = \sqrt {c^2 {\bf P}^2  + M^2 c^4 }$.  

Pryce, found that there are three possible definitions for the center of mass position vector.  However, only one of them is canonical and independent of the frame in which it is defined.  This is the natural and necessary choice if we want a theory that provides the same physics to all observers and is compatible with quantum mechanics. In our case, $\bf{X}$ is defined  in the $O$ frame by (see \cite{25}):
\beqn{\lb{cem}}
{\mathbf{X}} = \frac{1}{H}\sum\limits_{i = 1}^n {{H_i}{{\mathbf{x}}_i} + \frac{{{c^2}\left( {{\mathbf{S}} \times {\mathbf{P}}} \right)}}{{H\left( {M{c^2} + H} \right)}}}, 
\eeqn
where $\bf{S}$ is the global spin of the system of particles relative to $O$.   (It is clear that ({\rf{cem}}) cannot represent the vector part of a four-vector.)      If there is no interaction, $S, H$ and $M$ are constant, with no dependence on the $\{{\bf{x}}_i, \, {\bf{p}}_i\}$ variables, so that:
\[
\left\{ {{X_i},{X_j}} \right\} = \sum\limits_{k = 1}^n {\frac{{\partial {X_i}}}{{\partial {{\mathbf{p}}_k}}} \cdot \frac{{\partial {X_j}}}{{\partial {{\mathbf{x}}_k}}} - \frac{{\partial {X_j}}}{{\partial {{\mathbf{p}}_k}}} \cdot \frac{{\partial {X_i}}}{{\partial {{\mathbf{x}}_k}}}}  \equiv 0.
\]
However, when interaction is present, $S, H$ and $M$ may all depend on the $\{{\bf{x}}_i, \, {\bf{p}}_i\}$ variables, so that in general $\left\{ {{X_i},{X_j}} \right\} \ne 0$.  Since $\bf{X}$  is the canonical conjugate of $\bf{P}$, it precisely what we need for a consistent merge with quantum mechanics.

Let $\bf{V}$ be the velocity of $\bf{X}$ with respect to $O$.  It follows that $H$ also has the representation $H= Mc^2\g({\bf{V}})$, so that  $\g({\bf{V}})^{-1}=(Mc^2 /H)$.  In this representation, we see that  $d\tau =\g({\bf{V}})^{-1}dt = (Mc^2 /H)dt$ does not depend on the number of particles in the system.    It follows that, as long as $Mc^2 /H$ is fixed,  $\tau$ is invariant, so that the number of particles $n$, can increase or decrease without changing $\tau$.  (This means that number $n$ is not conserved and, in some cases of physical interest,  may even be a  integer-valued random variable).

From $d{t^2} = d{\tau ^2} + {{d{{\mathbf{X}}^2}} \mathord{\left/
 {\vphantom {{d{{\mathbf{X}}^2}} {{c^2}}}} \right.
 \kern-\nulldelimiterspace} {{c^2}}}$, we see that  (${\mathbf{U}} = {{d{\mathbf{X}}} \mathord{\left/
 {\vphantom {{d{\mathbf{X}}} {d\tau }}} \right.
 \kern-\nulldelimiterspace} {d\tau }}$)
\[
{c^2}d{t^2} = \left( {{c^2} +{\mathbf{U}}^2} \right)d{\tau ^2}\quad  \Rightarrow \quad cdt = \left( {\sqrt {{c^2} + {{\mathbf{U}}^2}} } \right)d\tau .
\]
It is easy to see that ${\mathbf{U}}=\g({\mathbf{V}}){\mathbf{V}}$, so that ${\mathbf{U}}$ is constant.
If we define $b=\sqrt{{\bf{U}}^2+c^2}$, we can write $cdt=bd{\tau}$.  Since $b$ is constant we have: $ct=b{\tau}$.  For observer $O'$  the same system has velocity $\bf{V'}$ for the center of mass and, by the same calculations, we obtain $ct'=b'{\tau}$, where $b'=\sqrt{{\bf{U'}}^2+c^2}$.  This shows that a unique (operational) measure of time is available to all observers.  Furthermore, $\tau$ differs from $t$ (respectively $t'$), by a constant scale factor.  Thus, all observers may uniquely define the  local-time of the center of mass for the system of particles (independent of their chosen inertial frame). We also obtain our third identity:
\beqn{\lb{id2}}
\frac{1}{c}\frac{d}{{dt}} \equiv \frac{1}{b}\frac{d}{{d\tau }}  \equiv \frac{1}{{{b_i}}}\frac{d}{{d{\tau _i}}} 
\eeqn 
Applying the above to ${\bf{x}}_i$ we see that:
\[
\frac{1}{c}\frac{{d{{\mathbf{x}}_i}}}{{dt}} \equiv \frac{1}{b}\frac{{d{{\mathbf{x}}_i}}}{{d\tau }} \equiv \frac{1}{{{b_i}}}\frac{{d{{\mathbf{x}}_i}}}{{d{\tau _i}}}.
\]
\begin{thm}If $O$ is observing any system of particles, there are two sets of global variables available: $({\bf{X}}, t)$ and $({\bf{X}}, \tau)$.  Use of $({\bf{X}}, t)$ provides a relative definition of time and a constant speed of light;  while use of $({\bf{X}}, \tau)$ provides a unique definition of time and a relative definition of the speed of light, with no upper bound.
\end{thm}
\begin{proof}The first part is clear.  To prove the second statement, from above, we see that any other observer $O'$ investigating the same system of particles also has two sets of global variables available: $({\bf{X'}}, t')$ with a constant speed of light and  $({\bf{X'}}, \tau)$ with $b'$ relative.
We are done if we can show that Einstein's first postulate holds.

Let ${\bf{W}}$ be the relative velocity between observer $O$ and $O'$.  Since $\tau$ is the same for both we only need the  relationship between the two scale factors $b$ and $b'$ to satisfy the first postulate.  Starting with $t' = \tf{b'}{c}\tau  = \gamma \left( {\mathbf{W}} \right)\left( {\tf{b}{c}\tau  - {{{\mathbf{X}} \cdot {\mathbf{V}}} \mathord{\left/
 {\vphantom {{{\mathbf{X}} \cdot {\mathbf{W}}} {{c^2}}}} \right.
 \kern-\nulldelimiterspace} {{c^2}}}} \right)$, we see that, since ${\mathbf{U}} = \left( {{{\mathbf{X}} \mathord{\left/
 {\vphantom {{\mathbf{X}} \tau }} \right.
 \kern-\nulldelimiterspace} \tau }} \right)$, we get: 
\[
b' = \gamma \left( {\mathbf{W}} \right)\left[ {b - \left( {{{\mathbf{X}} \mathord{\left/
 {\vphantom {{\mathbf{X}} \tau }} \right.
 \kern-\nulldelimiterspace} \tau }} \right) \cdot \left( {{{\mathbf{W}} \mathord{\left/
 {\vphantom {{\mathbf{W}} {{c}}}} \right.
 \kern-\nulldelimiterspace} {{c}}}} \right)} \right] = \gamma \left( {\mathbf{W}} \right)\left( {b - {{{\mathbf{U}} \cdot {\mathbf{W}}} \mathord{\left/
 {\vphantom {{{\mathbf{U}} \cdot {\mathbf{W}}} {{c}}}} \right.
 \kern-\nulldelimiterspace} {{c}}}} \right).
\]
A similar calculation shows that  $b = \gamma ({\mathbf{W}})\left( {b' + {{{\mathbf{U'}} \cdot {\mathbf{W}}} \mathord{\left/
 {\vphantom {{{\mathbf{U'}} \cdot {\mathbf{W}}} {{c}}}} \right.
 \kern-\nulldelimiterspace} {{c}}}} \right)$.  This shows that each observer can have direct access to all information available to any other observer once they know their relative velocities.  Thus the first postulate of Einstein is satisfied at the global level.
\end{proof}
\begin{cor} The two global sets of variables produce mathematically equivalent  theories, but do not produce physically equivalent  theories. 
\end{cor}
\begin{thm} The special theory of Einstein holds for any many-particle system and is independent of the Minkowski postulate. 
\end{thm}
This result is fundamental to our approach, since we do not require  the particle coordinates to transform as four-vectors.  Thus, the no-interaction theorem does not apply.  In the following section, we study the dynamics of the system.
\begin{rem}
This distinction may also prove important in the future, because there continues to appear  research in cosmology, applied physics and engineering, suggesting that the constant $c$ is  not an upper bound in all cases (see for example  \cite{36, 37, 38, 39}).

For many experiments (e. g., high energy particle studies) the center of mass is the natural frame of choice.  In this case, $t=\tau$ and one has a constant speed of light for all events associated with the global system of (interacting) particles.  (However, from  equation (2.3)  individual particles can still have velocities much larger than $c$.) 
\end{rem}

\section{General Dynamics}
In this section, we focus on the general dynamics of our system of particles.  The study of external and internal dynamics will be accomplished in Part II.
\subsection{Poincar\' e algebra}
If we let ${\bf L}$  be the generator of pure Lorentz transformations (boost) and define the total angular momentum ${\bf J}$ by   
\[ 
{\bf J} = \sum\limits_{i = 1}^n {{\bf x}_i  \times {\bf p}_i } ,   
\]  
we obtain the following Poisson Bracket relations characteristic of the Lie algebra for
the Poincar\' e group, when we use the time $t$ of our observer $O$:   
\beqn{\lb{cr}} 
\begin{gathered}
  \frac{{d{\mathbf{P}}}}{{dt}} = \left\{ {H,{\mathbf{P}}} \right\} = 0\quad \quad {\text{   }}\frac{{d{\mathbf{J}}}}{{dt}} = \left\{ {H,{\mathbf{J}}} \right\} = 0\quad \quad {\text{ }}\left\{ {{P_i},{P_j}} \right\} = 0 \hfill \\
  \left\{ {{J_i},{P_j}} \right\} = {\varepsilon _{ijk}}{P_k}\quad \quad {\text{   }}\left\{ {{J_i},{J_j}} \right\} = {\varepsilon _{ijk}}{J_k}\quad \quad {\text{ }}\left\{ {{J_i},{L_j}} \right\} = {\varepsilon _{ijk}}{L_k} \hfill \\
  \frac{{d{\mathbf{L}}}}{{dt}} = \left\{ {H,{\mathbf{L}}} \right\} =  - {\mathbf{P}}\quad \quad \left\{ {{P_i},{L_j}} \right\} = {\delta _{ij}}\frac{H}{{{c^2}}}\quad \quad \left\{ {{L_i},{L_j}} \right\} = {\varepsilon _{ijk}}\frac{{{J_k}}}{{{c^2}}}. \hfill \\ 
\end{gathered} 
\eeqn
It is easy to see that $M$ commutes with $H$, ${\bf P}$,  and ${\bf J}$, and to show
that $M$ commutes with ${\bf L}$. 

\subsubsection{Canonical Hamiltonian}
If we treat the system of particles as a single entity, then (${\bf{X}},{\bf{P}}$) are the natural phase space variables for the external system dynamics.  In this case, if $W({\bf{X}},{\bf{P}})$ is a dynamical parameter in $\bf{X}$ and $\bf{P}$,  the time evolution of $W$ is defined by:
\beqn
\frac{{dW}}{{dt}} = \left\{ {H,W} \right\} = \frac{{\partial H}}{{\partial {\mathbf{P}}}} \cdot \frac{{\partial W}}{{\partial {\mathbf{X}}}} - \frac{{\partial H}}{{\partial {\mathbf{X}}}} \cdot \frac{{\partial W}}{{\partial {\mathbf{P}}}}.
\eeqn
In order to represent the dynamics using the  proper time of the system, we use the representation $d\tau = ({{Mc^2 } \mathord{\left/ {\vphantom {{mc^2 } H}} \right. \kern-\nulldelimiterspace} H})dt$,
so that:
\[
\frac{{dW}}{{d\tau }} =\frac{{dt}}{{d\tau }} \frac{{dW}}{{dt}} = \frac{H}{{M{c^2}}}\left\{ {H,W} \right\} = \left( {\frac{H}{{M{c^2}}}\frac{{\partial H}}{{\partial {\mathbf{P}}}}} \right) \cdot \frac{{\partial W}}{{\partial {\mathbf{X}}}} - \left( {\frac{H}{{M{c^2}}}\frac{{\partial H}}{{\partial {\mathbf{X}}}}} \right) \cdot \frac{{\partial W}}{{\partial {\mathbf{P}}}}.
\]
The ratio ${H \mathord{\left/ {\vphantom {H {M{c^2}}}} \right.  \kern-\nulldelimiterspace} {M{c^2}}}$ is constant and $Mc^2$ is a well-defined (invariant) for the system, so we can determine the canonical Hamiltonian $K$, related to $\tau$ by:
\[
\left\{ {K,W} \right\} = \frac{H}{{Mc^2 }}\left\{ {H,W} \right\},\quad \left. K \right|_{{\mathbf{P}} = 0}  = \left. H \right|_{{\mathbf{P}} = 0}  = Mc^2. 
\]
In this case: 
\[
\begin{gathered}
  \left\{ {K,W} \right\} = \left[ {\frac{H}{{Mc^2 }}\frac{{\partial H}}{{\partial {\mathbf{P}}}}} \right]\frac{{\partial W}}{{\partial {\mathbf{X}}}} - \left[ {\frac{H}
{{Mc^2 }}\frac{{\partial H}}{{\partial {\mathbf{X}}}}} \right]\frac{{\partial W}}
{{\partial {\mathbf{P}}}} \hfill \\
  {\text{          }} = \frac{\partial }{{\partial {\mathbf{P}}}}\left[ {\frac{{H^2 }}
{{2Mc^2 }} + a} \right]\frac{{\partial W}}{{\partial {\mathbf{X}}}} - \frac{\partial }
{{\partial {\mathbf{X}}}}\left[ {\frac{{H^2 }}{{2Mc^2 }} + a'} \right]\frac{{\partial W}}
{{\partial {\mathbf{P}}}}, \hfill \\ 
\end{gathered} 
\]
we get that $a = a' = \tfrac{1}{2}Mc^2$, so that
\beqn{\lb{H}}
K = \frac{{{H^2}}}{{2M{c^2}}} + \frac{{M{c^2}}}{2} = \frac{{{{\mathbf{P}}^2}}}{{2M}} + M{c^2},\quad {\text{and}}\quad \frac{{dW}}{{d\tau }} = \left\{ {K,W} \right\}.
\eeqn 					 
Thus, $K$ looks like the standard (Newtonian) Hamiltonian except for the $M{c^2}$ term.
\subsubsection{Proper time Poincar\' e algebra}
We can use the same definitions for ${\bf P}$, ${\bf J}$, and ${\bf L}$  to
obtain our new commutation relations: 
\beqn\lb{cr'}
\begin{gathered}
  \frac{{d{\mathbf{P}}}}{{d\tau }} = \left\{ {K,{\mathbf{P}}} \right\} = {\mathbf{0}},\quad \;\;\;\;\frac{{d{\mathbf{J}}}}{{d\tau }} = \left\{ {K,{\mathbf{J}}} \right\} = {\mathbf{0}},\quad \;\left\{ {{P_i},{P_j}} \right\} = 0, \hfill \\
  \;\left\{ {{J_i},{P_j}} \right\} = {\varepsilon _{ijk}}{P_k},\quad \quad \left\{ {{J_i},{J_j}} \right\} = {\varepsilon _{ijk}}{J_k},\quad \quad \left\{ {{J_i},{L_j}} \right\} = {\varepsilon _{ijk}}{L_k}, \hfill \\
  \frac{{d{\mathbf{L}}}}{{d\tau }} = \left\{ {K,{\mathbf{L}}} \right\} = \tfrac{{ - H}}{{M{c^2}}}{\mathbf{P}},\quad \left\{ {{P_i},{L_j}} \right\} = \tfrac{{ - H}}{{{c^2}}}{\delta _{ij}},\quad \left\{ {{L_i},{L_j}} \right\} = \tfrac{{ - {J_k}}}{{{c^2}}}{\varepsilon _{ijk}}. \hfill \\ 
\end{gathered} 
\eeqn
We see again that, except for a (constant) change of scale, we obtain the same Lie algebra for the  Poincar\'{e} group.   Thus, the replacement of $t$ with $\tau$ still produces a relativistic theory.  We will explicitly construct and discuss the corresponding Lorentz transformations that fix $\tau$ in Part II.
\subsubsection{The $t \to \tau$ contact group}
The mapping between $t$ and  $\tau$ is a member of the family of contact groups, often used in celestial mechanics (see \cite{26}).  Contact transformations are sometimes called tangency transformations in mechanics, because they leave invariant the tangent at the point of contact. In what follows (from our identities) we use ${\raise0.5ex\hbox{$\scriptstyle {\mathbf{w}}$}
\kern-0.1em/\kern-0.15em
\lower0.25ex\hbox{$\scriptstyle c$}}$ with ${\gamma ^{ - 1}}$ and ${\raise0.5ex\hbox{$\scriptstyle {\mathbf{u}}$}
\kern-0.1em/\kern-0.15em
\lower0.25ex\hbox{$\scriptstyle b$}}$ with $\gamma$. In this case, an explicit representation is easy:
\[
\begin{gathered}
  d\tau  = {\gamma ^{ - 1}}dt = {{\gamma '}^{ - 1}}dt' \Rightarrow  \hfill \\
  dt = \gamma d\tau {\text{  and }}dt' = \gamma 'd\tau {\text{ }} \hfill \\ 
\end{gathered} 
\]
These transformations are clearly invertible.  Since the observer frames are inertial, we have that $t=\g \tau$ and $t' =\g' \tau$.  Thus the transformation $t \to \tau$ induces the contact mapping of ${\bf C}^{-1}[\,t,\,\tau]: ( O, t)\, \to \,(O, \tau)$. (See \cite{27} pg. 1312, for the general case.)

Let observer $O'$ with time $t'$ observe the same system of particles.  From this frame the total Hamiltonian is $H'$.  One can also construct $\bf{P'}$ and $M'$, leading to the same form for the commutation relations as in (\rf{cr'}).

Let the contact maps from $( O, \tau)\, \to \,(O, t)$ and from $( O', \tau)\, \to \,(O', t')$ be denoted by ${\bf
C}[\,t,\,\tau]$ and ${\bf C}[\,t',\,\tau]$ respectively.   Let  ${\mcP}(O', O)$ be the Poincar\' e map from $O \to O'$.
\begin{thm}  The system of particles as seen by an observer at $O$ is related to that of an observer at $O'$ by the Zachary transformation: 
\beqn
\begin{gathered}
O'({\bf{X}' },\, \tau)={\mathcal{Z}}[O',O, \tau ]O({\bf{X} },\, \tau)   \hfill \\ 
={\bf C}[ \, \tau,\,t' ]{\mcP}(O',\,O){\bf C}^{-1}[ \,t,\,\tau ]O({\bf{X} },\, \tau).  \hfill \\ 
\end{gathered}
\eeqn
\end{thm}
\begin{rem} The above transformation is named for our deceased colleague Woodford W. Zachary (see \cite{27}, equation (5.21)).  
\end{rem}
\begin{proof}The proof follows since the diagram below is
commutative.  
\[
 \begin{matrix} {O({\bf{X} },\, t)}  &  {\rm  {}} & { \xrightarrow{{\mathcal{P}}}} 
&  {\rm  {}} & O'({\bf{X}' },\, t') \cr &  {\rm  {}}&  {\rm  {}}&  {\rm  {}}& 
{\rm  {}}&  {\rm  {}}\cr &  {\rm  {}}&  {\rm  {}}&  {\rm  {}}&  {\rm  {}}&  {\rm 
{}}\cr {{{\bf C}^{-1}[\,t,\,\tau]}} &   \Bigg\uparrow &   {\rm  {}} & 
\Bigg\downarrow & {{\bf C}[\,t',\,\tau]}\cr &  {\rm  {}}&  {\rm  {}}&  {\rm  {}}& 
{\rm  {}}&  {\rm  {}}\cr &  {\rm  {}}&  {\rm  {}}&  {\rm  {}}&  {\rm  {}}&  {\rm 
{}}\cr O({\bf{X} },\, \tau)  &  {\rm  {}} &  \xrightarrow{\mcZ} {} &  {\rm  {}} &
O'({\bf{X}' },\, \tau)
\end{matrix} 
\]
\end{proof} 
Since $K$ does not depend on the center-of-mass position ${\bf X}$, it is easy to see
that  
\beqn{\lb{X}} 
{\mathbf{U}}=\frac{{d{\mathbf{X}}}}{{d\tau }} = \frac{{\partial K}}{{\partial {\mathbf{P}}}} = \frac{{\mathbf{P}}}{M} = \frac{1}{M}\sum\limits_{i = 1}^n {{{\mathbf{p}}_i}}  = \frac{1}{M}\sum\limits_{i = 1}^n {{m_i}{{\mathbf{u}}_i}}  = \frac{1}{M}\sum\limits_{i = 1}^n {{m_i}\frac{{d{{\mathbf{x}}_i}}}{{d\tau_i }}}. \; 
\eeqn
For the $O'$ observer, the same calculation leads to:
\beqn{\lb{X'}}
\quad \quad  {\mathbf{U}}'=\frac{{d{\mathbf{X'}}}}{{d\tau }} = \frac{{\partial K}}{{\partial {\mathbf{P'}}}} = \frac{{{\mathbf{P'}}}}{{M'}} = \frac{1}{{M'}}\sum\limits_{i = 1}^n {{{{\mathbf{p}'_i}}}}  = \frac{1}{{M'}}\sum\limits_{i = 1}^n {{{m'_i}}{{{\mathbf{u}'_i}}}}  = \frac{1}{{M'}}\sum\limits_{i = 1}^n {{{m'_i}}\frac{{d{{{\mathbf{x}'_i}}}}}{{d\tau_i }}}. 
\eeqn 
We now observe that
 \[
dt = \frac{{{H_i}}}{{{m_i}{c^2}}}d{\tau _i} = \frac{H}{{M{c^2}}}d\tau \quad  \Rightarrow \quad \frac{{{m_i}}}{M}\frac{d}{{d{\tau _i}}} = \frac{{{H_i}}}{H}\frac{d}{{d\tau }}.
 \]
Thus, we can replace ({\rf{X}}) and ({\rf{X'}}) by: 
\beqn
\frac{{d{\mathbf{X}}}}{{d\tau }} = \frac{{\partial K}}{{\partial {\mathbf{P}}}} = \frac{{{\mathbf{P}}}}{{M}} = \frac{1}{{M}}\sum\limits_{i = 1n}^n {{{{\mathbf{p}}}_i}}  = \frac{1}{{M}}\sum\limits_{i = 1}^n {{{m_i}}{{{\mathbf{u}}_i}}}  = \frac{1}{H}\sum\limits_{i = 1}^n {{H_i}\frac{{d{{\mathbf{x}}_i}}}{{d\tau }}} 
\eeqn
and
\beqn 
\frac{{d{\mathbf{X'}}}}{{d\tau }} = \frac{{\partial K}}{{\partial {\mathbf{P'}}}} = \frac{{{\mathbf{P'}}}}{{M'}} = \frac{1}{{M'}}\sum\limits_{i = 1}^n {{{{\mathbf{p'}}}_i}}  = \frac{1}{{M'}}\sum\limits_{i = 1}^n {{{m'_i}}{{{\mathbf{u}'_i}}}}  = \frac{1}{{H'}}\sum\limits_{i = 1n}^n {{{H'_i}}\frac{{d{{{\mathbf{x}'_i}}}}}{{d\tau }}}. 
\eeqn
Since the $H_i$ (respectively $H'_i$) do not depend on $\tau$, we can integrate both equations to get:
\[
{\mathbf{X}} = \frac{1}{H}\sum\limits_{i = 1n}^n {{H_i}{{\mathbf{x}}_i} + {\mathbf{Y}}} \quad {\text{and}}\quad {\mathbf{X'}} = \frac{1}{{H'}}\sum\limits_{i = 1n}^n {{{H'_i}}{\mathbf{x'_i}} + {\mathbf{Y'}}} ,
\]
where ${\mathbf{Y}}$ and ${\mathbf{Y'}}$ are constants of integration.  (This shows that the canonical Hamiltonian determines the canonical position up to a constant.)

To see directly that the clock transformation is also a canonical change of variables (time), which leaves phase space invariant, we have the following theorem. 
\begin{thm}
There exists a function $S = S\left( {{\bf{X}}  ,\; {\bf{P}}  ,\;\tau } \right)$ such that
\[
{\mathbf{P}} \cdot d{\mathbf{X}} - Hdt \equiv {\mathbf{P}} \cdot d{\mathbf{X}} - Kd\tau  + dS. 
\]
\end{thm}
\begin{proof}
Set $S = [Mc^2  - K]\tau$.   An easy calculation, using the fact that both $Mc^2$ and $K$ are conserved quantities, shows that $dS=[Mc^2  - K]d{\tau}$.  An additional easy calculation gives the  result.
\end{proof}
We note that
\[
\sum\limits_{i = 1}^n {\left[ {{{\mathbf{p}}_i} \cdot d{{\mathbf{x}}_i} - {H_i}dt} \right]}  = \sum\limits_{i = 1}^n {{{\mathbf{p}}_i} \cdot d{{\mathbf{x}}_i}}  - \sum\limits_{i = 1}^n {{H_i}dt}  = \sum\limits_{i = 1}^n {{{\mathbf{p}}_i} \cdot d{{\mathbf{x}}_i}}  - Hdt.
\]
This result  and $dS=[Mc^2  - K]d{\tau}$ is sufficient to justify the following:
\begin{cor} 
There exists a function $S = S\left( {\{ {\bf{x}}_i \} ,\;\{ {\bf{p}}_i \} ,\;\tau } \right)$ such that
\[
\sum\limits_{i = 1}^n {{{\mathbf{p}}_i} \cdot d{{\mathbf{x}}_i}}  - Hdt \equiv \sum\limits_{i = 1}^n {{{\mathbf{p}}_i} \cdot d{{\mathbf{x}}_i}}  - Kd\tau  + dS.
\]
\end{cor}
\begin{Def} A theory is said to be Einsteinian if at least one representation exists, which satisfies the two postulates of the special theory. 
\end{Def}
\begin{thm} Any closed system of interacting particles is Einsteinian and independent of the Minkowski postulate. Furthermore, there always exists two  distinct sets of  inertial frame coordinates  for each observer,  to describe each particle in the system and the system as a whole. The following holds: 
\begin{enumerate}
\item In one frame, the speed of light is an invariant upper bound and time is relative, while in the other, time is invariant and the speed of light $b$, is relative with no upper bound.
\item For the whole system and for each particle, the equations of motion are mathematically equivalent.
\end{enumerate}
\end{thm}
We have already proven all but the last part of the above theorem.  The next section is devoted to external dynamics.  We complete our proof in the second part, when we study electrodynamics. 
\subsection{ General System Dynamics}
In this section we view the system from an external perspective as if it is one interacting particle.  At this level,  it suffices to assume the interaction is via a potential $V({\bf{X}})$.
We can add $V$ to the equation for $H$, to get:
\beqn\lb{eg2}
\begin{gathered}
  H = \sqrt {{c^2}{{\mathbf{P}}^2} + {M^2}{c^4}}  + V({\mathbf{X}}) = {H_0} + V({\mathbf{X}}) \Rightarrow  \hfill \\
  \frac{{d{\mathbf{X}}}}{{dt}} = \frac{{{c^2}{\mathbf{P}}}}{{{H_0}}}\quad {and}\quad \frac{{d{\mathbf{P}}}}{{dt}} =  - \nabla V({\mathbf{X}}). \hfill \\ 
\end{gathered} 
\eeqn
For comparison, if we use the proper clock, we get: 
\beqn{\lb{eg3}}
\begin{gathered}
  K = \frac{{{H^2}}}{{2M{c^2}}} + \frac{{M{c^2}}}{2} \Rightarrow \frac{{d{\mathbf{X}}}}{{d\tau }} = \frac{{\partial K}}{{\partial {\mathbf{P}}}} = \frac{H}{{M{c^2}}}\frac{{{c^2}{\mathbf{P}}}}{{{H_0}}} = \frac{b}{c}\frac{{d{\mathbf{X}}}}{{dt}}, \hfill \\
  \frac{{d{\mathbf{P}}}}{{d\tau }} = \frac{{\partial K}}{{\partial {\mathbf{X}}}} =  - \frac{H}{{M{c^2}}}\nabla V\left( {\mathbf{X}} \right) = \frac{b}{c}\frac{{d{\mathbf{P}}}}{{dt}}. \hfill \\ 
\end{gathered}
 \eeqn
Comparison of ({\rf{eg2}}) and ({\rf{eg3}}) shows that the two clocks give mathematically equivalent equations of motion for the general system dynamics.
\subsection{The Clock Relationship}
There is a basic relationship between the global system clock and the clocks of the individual particles.  To derive this relationship, return to our definition of the global Hamiltonian $K$ and let $W$ be
any observable.  Then  
\begin{align}{\lb{c1}}
  & {{dW} \over {d\tau }} = \left\{ {K,W} \right\} = {H \over {Mc^2 }}\left\{ {H,W}
\right\} =  {H \over {Mc^2 }}\sum\limits_{i = 1}^n {\left\{ {H_i ,W} \right\}}   \cr 
  & \qquad = {H \over {Mc^2 }}\sum\limits_{i = 1}^n {{{m_i c^2 } \over {H_i
}}\left[ {{{H_i } \over {m_i c^2 } }\left\{ {H_i ,W} \right\}} \right]}  =
\sum\limits_{i = 1}^n {{{Hm_i } \over {MH_i }}\left\{ {K_i ,W} \right\}}.  
\end{align} 
Using the (easily derived) fact that $d\tau _i /d\tau  = Hm_i
/MH_i  = b_i /b$, we get 		  
\beqn{\lb{c2}}
{{dW} \over {d\tau }} = \sum\limits_{i = 1}^n
{{{d\tau _i } \over {d\tau }}\left\{ {K_i ,W} \right\}}. 
\eeqn           
Equation ({\rf{c2}}) allows us to relate the global system dynamics to the local systems dynamics.  Let us combine equations ({\rf{c1}}) and ({\rf{c2}}), to get our third identity: 
\beqn{\lb{id3}}
d\tau \left\{ {K,W} \right\} \equiv  \sum\limits_{i = 1}^n {d{\tau _i}\left\{ {{K_i},W} \right\}} \quad  \Rightarrow \quad d\tau {K^P}  = \sum\limits_{i = 1}^n d{\tau _i} {K_i^P}.
\eeqn
Where the last equation is strictly defined with the Poisson brackets.  This  provides 
 the basis for a many particle relativistic quantum theory with a universal wave function, using the transition to Heisenberg brackets on both sides (geometric quantization).
 
In closing this part, we recall that, in some cases, it is natural to place the observer at the center of mass.  In this case,  equation ({\rf{c2}}) can be written as:
\beqn{\lb{c4}}
{{dW} \over {dt }} = \sum\limits_{i = 1}^n
{{{d\tau _i } \over {dt }}\left\{ {K_i ,W} \right\}} 
\eeqn           
and equation ({\rf{id3}}) can be written as:
\beqn{\lb{id5}}
dt \left\{ {H,W} \right\} \equiv  \sum\limits_{i = 1}^n {d{\tau _i}\left\{ {{K_i},W} \right\}} \quad  \Rightarrow \quad dt {H^P}  = \sum\limits_{i = 1}^n d{\tau _i} {K_i^P}.
\eeqn
\section{\bf Maxwell and Einstein without Minkowski}
If we replace $t$ by $\tau$ at the global level for electrodynamics, no new results are produced other then what is expected from Part I.  Thus, we focus on the direct interaction of a particle with an external field,  another particle or the local interaction of particles as seen from the center of mass.  
\subsection{{Maxwell Particle Dynamics}}
\subsubsection{ Dynamics of a Particle }  
We now investigate the corresponding single particle dynamical  theory. In this section, $b=b_i$,  $\tau =\tau_i$ and ${\bf{u}}={\bf{u}}_i$. 

Since $\tau$ is invariant during interaction (minimal coupling), we make the natural assumption that the form of $K$ also remains invariant.  Thus, if $\sqrt {c^2 {\mathbf{p}}^2  + m^2 c^4 }  \to \sqrt {c^2 {\bs{\pi }}^2  + m^2 c^4 }  + V$, where  $\bf A$ is the vector potential, $V=e\Phi$ is the potential energy, ${\mathbf{E}} =  - \tfrac{1}{b}\left( {\partial {\mathbf{A}}/\partial \tau } \right) - \bs{\nabla} \Phi $ and  ${\bs{\pi }} = {\mathbf{p}} - \tfrac{e}
{c}{\mathbf{A}}$.  In this case, $K$ becomes:
\[
K = \frac{{{H^2}}}{{2m{c^2}}} + \frac{{m{c^2}}}{2}=\frac{{ {\bs{\pi }}^2}}
{{2m}} + mc^2  + \frac{{V^2 }}
{{2mc^2 }} + \frac{{V\sqrt {c^2  {\bs{\pi }}^2+ m^2 c^4 } }}
{{mc^2 }}.
\]
If we set $H_0=\sqrt {c^2 {\bs{\pi }}^2  + m^2 c^4 }$, use standard vector identities with ${\bs{\nabla}} \times \bs\pi=-\tfrac{e}{c}\bf{B}$, and compute Hamilton's equations, we get: 
\beqn{\lb{ipd1}}
\frac{{d{\mathbf{x}}}}{{d\tau }} = \frac{{\partial K}}{{\partial {\mathbf{p}}}} = \frac{H}{{m{c^2}}}\left( {\frac{{{c^2}{\bs{\pi }} }}{{{H_0}}}} \right) = \frac{b}{c}\left( {\frac{{{c^2}{\bs{\pi }} }}{{{H_0}}}} \right) = \frac{b}{c}\frac{{d{\mathbf{x}}}}{{dt}}
\eeqn
and
\beqn{\lb{ipd2}}
\begin{gathered}
  \frac{{d{\mathbf{p}}}}{{d\tau }} = \frac{b}{c}\frac{{\left[ {\left( {{c^2}{\bs{\pi }}  \cdot {\bs{\nabla}} } \right){\mathbf{A}} + \tfrac{e}{b}\left( {{c^2}{\bs{\pi }}  \times {\mathbf{B}}} \right)} \right]}}{{{H_0}}} - \frac{b}{c}{\bs{\nabla}} V \hfill \\
   = \frac{b}{c}\left[ {\left( {{\mathbf{u}} \cdot {\bs{\nabla}} } \right){\mathbf{A}} + \tfrac{e}{b}\left( {{\mathbf{u}} \times {\mathbf{B}}} \right)} \right] - \frac{b}{c}{\bs{\nabla}} V \hfill \\
   = \frac{b}{c}\left[ {e{\mathbf{E}} + \tfrac{e}{b}\left( {{\mathbf{u}} \times {\mathbf{B}}} \right) + \tfrac{e}{b}\frac{{d{\mathbf{A}}}}{{d\tau }}} \right]\quad  \Rightarrow  \hfill \\
\frac{c}{b}\frac{{d{\bs\pi} }}{{d\tau }}
= \left[ {e{\mathbf{E}} + \tfrac{e}{b}\left( {{\mathbf{u}} \times {\mathbf{B}}} \right)} \right]=\frac{{d{\bs\pi} }}{{dt }}. \hfill \\ 
\end{gathered} 
\eeqn
Equations ({\rf{ipd1}}) and ({\rf{ipd2}}) show that the standard and dual equations of motion are mathematically equivalent.  Thus, our assumption that $K$ remain invariant with minimal coupling was the correct choice.   This also completes the proof of Theorem 3.6.
\subsubsection{{Field of a Particle}}  
To study the field of a particle, we write Maxwell's equations (in c.g.s. units): 
\beqn{\lb{imx1}}
\begin{gathered}
\nabla  \cdot {\mathbf{B}} = 0,\quad \quad \quad \nabla  \cdot {\mathbf{E}} = 4\pi \rho , \hfill \\
\nabla  \times {\mathbf{E}} =  - \frac{1}{c}\frac{{\partial {\mathbf{B}}}}{{\partial t}},\quad \nabla  \times {\mathbf{B}} = \frac{1}{c}\left[ {\frac{{\partial {\mathbf{E}}}}
{{\partial t}} + 4\pi \rho {\mathbf{w}}} \right]. \hfill \\ 
\end{gathered} 
\eeqn
Using equations ({\rf{id1}}) and ({\rf{id1a}}) in ({\rf{imx1}}), we have ({{\it{the mathematically identical  representation}}):
\beqn{\lb{imx2}}
\begin{gathered}
\nabla  \cdot {\mathbf{B}} = 0,\quad \quad \quad \nabla  \cdot {\mathbf{E}} = 4\pi \rho , \hfill \\
\nabla  \times {\mathbf{E}} =  - \frac{1}{b}\frac{{\partial {\mathbf{B}}}}{{\partial \tau }},\quad \nabla  \times {\mathbf{B}} = \frac{1}{b}\left[ {\frac{{\partial {\mathbf{E}}}}
{{\partial \tau }} + 4\pi \rho {\mathbf{u}}} \right]. \hfill \\ 
\end{gathered} 
\eeqn
Thus, we obtain a mathematically equivalent set of Maxwell's equations using the local time of the particle to describe its fields.  

To derive the corresponding wave equations, we  take the curl of the last two equations in ({\rf{imx2}}), and use standard vector identities, to get: 
\beqn{\lb{imx3}}
\begin{gathered}
 \frac{1}{{b^2 }}\frac{{\partial^2 {\mathbf{B}}}}{{\partial \tau ^2 }} - \frac{{{\mathbf{u}} \cdot {\mathbf{a}}}}{{b^4 }}\left[ {\frac{{\partial {\mathbf{B}}}}
{{\partial \tau }}} \right] - \nabla ^2 \cdot {\mathbf{B}} = \frac{1}{b}\left[ 4\pi \nabla  \times (\rho {\mathbf{u}}) \right], \hfill \\
 \frac{1}{{b^2 }}\frac{{\partial ^2 {\mathbf{E}}}}{{\partial \tau ^2 }} - \frac{{{\mathbf{u}} \cdot {\mathbf{a}}}}{{b^4 }}\left[ {\frac{{\partial {\mathbf{E}}}}
{{\partial \tau }}} \right] - \nabla ^2  \cdot {\mathbf{E}} =  - \nabla (4\pi \rho ) - \frac{1}{b}\frac{\partial }{{\partial \tau }}\left[ {\frac{{4\pi (\rho {\mathbf{u}})}}{b}} \right], \hfill \\ 
\end{gathered} 
\eeqn				
where ${\bf{a}} = d{\bf{u}}/d\tau$ is the particle acceleration.  Thus, a new term arises when the proper-time of the charge is used to describe its fields.   This makes it clear that the local clock encodes information about the particle's interaction that is unavailable when the clock of the observer, co-moving observer or the proper clock of the center of mass  is used to describe the fields.   The new term in equation ({\rf{imx3}}) is dissipative, acts to oppose the acceleration, is zero when ${\bf{a}} =0$  or perpendicular to $\bs{u}$.  It also arises instantaneously with the force.  Furthermore,  this term does not depend on the nature of the force.  This is exactly what one expects of the back reaction caused by the inertial resistance of the particle to accelerated motion and, according to Wheeler and Feynman \cite{13}, is precisely what is meant by radiation reaction.  
\begin{rem}
It is of particular interest that this implies a charged particle can distinguish between inertial  and accelerating frames.   Thus, an observer in an elevator can always determine the state of motion.
From this point of view, it is no surprise that the 2.7 $^{\circ}$K MBR represents a unique preferred frame of rest.
\end{rem}
  If we make a scale transformation (at fixed position) with ${\bf{E}} \to (b/c)^{1/2}{\bf{E}}$    and ${\bf{B}} \to (b/c)^{1/2}{\bf{B}}$, the equations in ({\rf{imx3}}) transform to 
\beqn{\lb{imx4}}
\begin{gathered}
 \frac{1}{{b^2 }}\frac{{\partial ^2 {\mathbf{B}}}}{{\partial \tau ^2 }} - {\text{ }}\nabla ^2 {\kern 1pt}  \cdot {\mathbf{B}} + \left[ \frac{{\ddot b}}
{{2b^3 }}-{\frac{{3\dot b^2 }}{{4b^4 }}  } \right]{\mathbf{B}} = \frac{{c^{1/2} }}{{b^{3/2} }}\left[ {4\pi \nabla  \times (\rho {\mathbf{u}})} \right], \hfill \\
 \frac{1}{{b^2 }}\frac{{\partial ^2 {\mathbf{E}}}}{{\partial \tau ^2 }} - {\text{ }}\nabla ^2 {\kern 1pt}  \cdot {\mathbf{E}} +  \left[ \frac{{\ddot b}}
{{2b^3 }}-{\frac{{3\dot b^2 }}{{4b^4 }}  } \right]{\mathbf{E}} =  - \frac{{c^{1/2} }}{{b^{1/2} }}\nabla (4\pi \rho ) - \frac{{c^{1/2} }}{{b^{3/2}}}\frac{\partial}{{\partial \tau }}\left[ {\frac{{4\pi (\rho {\mathbf{u}})}}{b}} \right]. \hfill \\ 
\end{gathered} 
\eeqn
This is the Klein-Gordon equation with an effective mass $\mu$ given by
\beqn
\mu  = \left\{ {\frac{{\hbar ^2 }}{{c^2 }}\left[ {\frac{{\ddot b}}
{{2b^3 }} - \frac{{3\dot b^2 }}{{4b^4 }}} \right]} \right\}^{1/2}  = \left\{ {\frac{{\hbar ^2 }}{{c^2 }}\left[ {\frac{{{\mathbf{u}} \cdot {\mathbf{\ddot u}} + {\mathbf{\dot u}}^2 }}  {{2b^4 }} - \frac{{5\left( {{\mathbf{u}} \cdot {\mathbf{\dot u}}} \right)^2 }}
{{4b^6 }}} \right]} \right\}^{1/2}. 
\eeqn
In the following sections, we verify that our view of $\mu$ as an effective mass is correct.
\subsection{Radiation from the Accelerated  Charge }
In this section, we directly compute the radiation from an accelerated
charge.  Using potentials, it easy to check that, with  the Lorentz condition and
\beqn
{\mathbf{B}} = \nabla  \times {\mathbf{A}},\quad {\mathbf{E}} =  - \frac{1}{b}\frac{{\partial {\mathbf{A}}}}{{\partial \tau }} - \nabla \Phi ,
\eeqn
the wave equations for the potentials are: 
\beqn{\lb{imp1}}
\begin{gathered}
  \frac{1}{{{b^2}}}\frac{{{\partial ^2}{\mathbf{A}}}}{{\partial {\tau ^2}}} - \frac{{\left( {{\mathbf{u}} \cdot {\mathbf{a}}} \right)}}{{{b^4}}}\frac{{\partial {\mathbf{A}}}}{{\partial \tau }} - {\nabla ^2}{\mathbf{A}} = \frac{{4\pi \rho {\mathbf{u}}}}{b} \hfill \\
  \frac{1}{{{b^2}}}\frac{{{\partial ^2}\Phi }}{{\partial {\tau ^2}}} - \frac{{\left( {{\mathbf{u}} \cdot {\mathbf{a}}} \right)}}{{{b^4}}}\frac{{\partial \Phi }}{{\partial \tau }} - {\nabla ^2}\Phi  = 4\pi \rho  \hfill \\ 
\end{gathered} 
\eeqn
We could solve the equations ({\rf{imp1}}), but it is easier to first obtain the solution using the proper-time of the observer and
then transform our result to the proper-time of the source. This makes the computations easier to follow and also provides the result quicker.  We follow the approach due to Panofsky and Phillips \cite{34}.   In this regard, $\left( {{\bf x}(t),t}
\right)$ represent the field position and $\left( {{\bf x}'(t'),t'} \right)$ represent the
retarded position of a point charge source $q$, with  ${\bf r}={\bf x}-{\bf x}'$,
${d{\bf r}}/d{t'}=-{\bf w}$, and ${d^2{\bf r}}/d{t'}^2={\dot{\bf w}}$.  The solutions are the standard Lienard-Wiechert potentials,  given by 
\beqn{\lb{imp2}}
{\bf A}={{q{\bf w}} \over {cs}},\ \ \ \ \ \Phi ={q \over s},\ \ \ \ 
s=r-\left( {{{{\bf r}\cdot {\bf w}} \over c}} \right).  
\eeqn 
We obtain the  proper-time form by replacing ${\bf w}/c$ by ${\bf u}/b$ to get 
\beqn{\lb{imp3}}
{\bf A}={{q{\bf u}} \over {bs}},\ \ \ \ \
\Phi ={q \over s},\ \ \ \ s=r-\left( {{{{\bf r}\cdot {\bf u}} \over b}}\right). 
\eeqn
The source-point  and field variables are related by the condition  
\beqn{\lb{imp4}}
 r=\left| {{\bf
x}-{\bf x'}} \right|=c(t-t').
\eeqn 
In the proper time variables, $d{\bf r}/d\tau '=-{\bf u}=-d{\bf x}'/d\tau '$ and $\tau'$ is the retarded
proper-time of the source.  The corresponding {\bf E} and {\bf B} fields are
computed using equation ({\rf{imp1}}) in the form 
\beqn 
{\bf E}({\bf x},\tau )={-{1 \over {\bar b}}{{\partial {\bf A}({\bf x},\tau
)}\over {\partial\tau }}-\nabla \Phi({\bf x},\tau ),}\,\,\,{\bf B}({\bf x},\tau )=\nabla
\times {\bf A}({\bf x},\tau ) 
\eeqn 
with ${\bar {\bf u}}=d{\bf x}/d{\tau}$, where $\tau$ is the proper-time of the
present position of the source and $\bar b=\left( {\bar {\bf u}^2+c^2} \right)^{1/2}$. 
To compute the fields from the potentials, we observe that the components of the
$\nabla$ operator are partials at constant $\tau$, and therefore are not partials at constant
$\tau'$.  Also, the partial derivatives with respect to $\tau$ imply constant ${\bf x}$
and hence refer to the comparison of potentials at a given point over an interval in
which the coordinates of the source will have changed.  Since only variations in time 
with respect to $\tau'$ are given, we must transform $(\partial /\partial \tau )\left|
{_{\bf x}}\right.$ and $\nabla
\left| {_\tau } \right.$ to expressions in terms of ${\partial/{\partial \tau' }}\left| {_{\bf
x}} \right.$.  For this, we must first transform $({\rf{imp4}})$ into a relationship between
$\tau$ and $\tau'$.   The  required transformation is  
\beqn 
c(t-t')=\int_{\tau' }^\tau {b(s)ds}.
\eeqn
The best approach is to first relate ${\partial / {\partial t}}\left|{_{\bf x}}\right.$   to   ${\partial /
{\partial t'}} \left|{_{\bf x}}\right.$ and then convert them to relationships between
${\partial/{\partial \tau }}{\left|{_{\bf x}}\right.}$ and 
${\partial/{\partial\tau'}}{\left|{_{\bf x}}\right.}$.  This leads to (see \cite{34}, pg. 298):   
\beqn
{{\partial r} \over {\partial t'}}=-{{{\bf r}\cdot {\bf w}} \over r},\ \
{{\partial r} \over {\partial t}}=c\left( {1-{{\partial t'} \over {\partial t}}}
\right)={{\partial r} \over {\partial t'}}\cdot {{\partial t'} \over {\partial t}}=-{{{\bf
r}\cdot {\bf w}} \over r}{{\partial t'} \over {\partial t}}.
\eeqn
Since ${{\partial \tau }/{\partial t}}={c /b}$, we get:  
\beqn 
{{\partial r} \over {\partial
t}}=c{\partial  \over {\partial t}}\left( {t-t' } \right)={{\partial \tau } \over {\partial
t}}{\partial  \over {\partial \tau }}\int_{\tau' }^\tau  {b(s)ds}={c \over {\bar b}}\left[
{\bar b-b{{\partial \tau' } \over {\partial \tau }}} \right].
\eeqn   
We also have, using ${{\partial \tau' }/{\partial t'}}={c/b}$ , that  
\beqn 
{{\partial r} \over
{\partial t'}}={{\partial r} \over {\partial \tau' } }{{\partial \tau' } \over {\partial
t'}}={c \over b}{{\partial r} \over {\partial \tau' }}\Rightarrow {1 \over b}{{\partial
r} \over {\partial \tau' }}=-{{{\bf r}\cdot {\bf w}} \over {rc}}=-{{{\bf r}\cdot {\bf
u}} \over {rb}},
\eeqn
so ${{\partial r} /{\partial\tau' }}=-{{{\bf r}\cdot{\bf u}}/ r}$ and hence 
\beqn 
{{\partial r} \over {\partial t}}={{\partial r} \over {\partial \tau
}}{c \over {\bar b}}={c \over {\bar b} }\left[ {\bar b-b{{\partial \tau' } \over {\partial
\tau }}} \right]\Rightarrow {{\partial r} \over {\partial \tau }}=\left[ {\bar
b-b{{\partial \tau' } \over {\partial \tau }}} \right],
\eeqn 
\beqa 
{{\partial r} \over
{\partial \tau }}={{\partial r} \over {\partial \tau' }}{{\partial \tau' } \over {\partial
\tau }}=-{{{\bf r}\cdot {\bf u}} \over r}{{\partial \tau' } \over {\partial \tau
}}\Rightarrow \ \ \ -{{{\bf r}\cdot {\bf u}} \over r}{{\partial \tau' } \over {\partial \tau
}}=\left[ {\bar b-b{{\partial \tau' } \over {\partial \tau }}} \right].
\eeqa 
If we solve the above for ${{\partial \tau' }/ {\partial \tau }}$, we have: 
\beqn{\lb{rtd}} 
{{\partial \tau' }
\over {\partial \tau }}={{\bar b} \over b}{r \over s},\ \ \ s=r-{{{\bf r}\cdot {\bf u}}
\over b}.
\eeqn  
Using this, we see that 
\beqn{\lb{rt}}
{1 \over {\bar b}}{\partial  \over
{\partial \tau }}={1 \over b}\cdot {r \over s}{\partial  \over {\partial \tau'}}.
\eeqn   
From $\nabla r=-c\nabla t'=\nabla_1r+({{\partial r}/ {\partial
t'}}){\nabla t'}$, we see that 
\[
\nabla r={{\bf r} \over r}-{c \over b}\cdot {{{\bf
r}\cdot {\bf u}} \over r}\nabla t'\ \ \Rightarrow -c\nabla t'={{\bf r} \over r}-{c \over
b}\cdot {{{\bf r}\cdot {\bf u}} \over r}\nabla t'.
\]  
Using $c\nabla
t'=b\nabla \tau' $ and solving for $\nabla \tau' $, we get $\nabla \tau' =-\left( {{{\bf r}
\mathord{\left/ {\vphantom {{\bf r} {bs}}} \right. \kern-\nulldelimiterspace} {bs}}}
\right)$, so that 
\[
 \nabla =\nabla _1-{{\bf r} \over {bs}}\cdot {\partial \over {\partial
\tau' }}. 
\]    
We now compute ${\nabla _1s}$ and  ${\partial s}/{\partial
\tau' }$.  The calculations are easy, so we simply state the results: 
\[ 
\nabla _1s={{\bf
r} \over r}-{{\bf u} \over b}={1 \over r}\left( {{\bf r}-{{r{\bf u}} \over b}}  \right),
\] 
\[ 
{{\partial s} \over {\partial \tau' }}={{{\bf u}^2} \over b}-{{{\bf
r}\cdot {\bf u}} \over r}-{{{\bf r}\cdot {\bf a} } \over b}+{{\left( {{\bf r}\cdot {\bf
u}} \right)\left( {{\bf u}\cdot {\bf a}} \right)} \over {b^3}}.
\]
We can now calculate the fields.  The computations are long but follow those of
 \cite{34}, so we only record a few selected results.  We obtain  
\beqn
\begin{gathered}
   - \nabla \Phi  = \frac{q}{{{s^2}}}\nabla s = \frac{q}{{{s^2}}}\left( {{\nabla _1}s - \frac{{\mathbf{r}}}{{bs}}\frac{{\partial s}}{{\partial \tau }}} \right)\quad  \Rightarrow  \hfill \\
   - \nabla \Phi  = \frac{{q\left[ {{\mathbf{r}}\left( {1 - {{{{\mathbf{u}}^2}} \mathord{\left/
 {\vphantom {{{{\mathbf{u}}^2}} {{b^2}}}} \right.
 \kern-\nulldelimiterspace} {{b^2}}}} \right) - {\mathbf{u}}sb} \right]}}{{{s^3}}} + \frac{{q{\mathbf{r}}\left( {{\mathbf{r}} \cdot {\mathbf{a}}} \right)}}{{{b^2}{s^3}}} - \frac{{q{\mathbf{r}}\left( {{\mathbf{r}} \cdot {\mathbf{u}}} \right)\left( {{\mathbf{u}} \cdot {\mathbf{a}}} \right)}}{{{b^4}{s^3}}}. \hfill \\ 
\end{gathered} 
\eeqn
Now  use equation $({\rf{rt}})$ to get  
$$
 -{1 \over {\bar b}}{{\partial {\bf A}} \over {\partial \tau }}=\left( {-{1 \over b}}
\right)\left( {{r \over s} } \right){{\partial {\bf A}} \over {\partial \tau' }}\Rightarrow
$$  
\beqa 
- \frac{1}{b}\frac{{\partial {\mathbf{A}}}}{{\partial \tau }} = \frac{{ - \left( {q{\mathbf{r}}{u \mathord{\left/
 {\vphantom {u b}} \right.
 \kern-\nulldelimiterspace} b}} \right)\left[ {\left( {{{\mathbf{r}} \mathord{\left/
 {\vphantom {{\mathbf{r}} r}} \right.
 \kern-\nulldelimiterspace} r} - {{\mathbf{u}} \mathord{\left/
 {\vphantom {{\mathbf{u}} b}} \right.
 \kern-\nulldelimiterspace} b}} \right) \cdot \left( {{{\mathbf{u}} \mathord{\left/
 {\vphantom {{\mathbf{u}} b}} \right.
 \kern-\nulldelimiterspace} b}} \right)} \right]}}{{{s^3}}} + \frac{{ - q{r^2}{\mathbf{a}} + qr\left[ {{\mathbf{r}} \times \left( {{\mathbf{a}} \times {{\mathbf{u}} \mathord{\left/
 {\vphantom {{\mathbf{u}} b}} \right.
 \kern-\nulldelimiterspace} b}} \right)} \right]}}{{{b^2}{s^3}}} + \frac{{q{\mathbf{u}}\left( {{\mathbf{r}} \cdot {\mathbf{r}}} \right)\left( {{\mathbf{u}} \cdot {\mathbf{a}}} \right)}}{{{b^4}{s^3}}}.
\eeqa
Combining the above with $(6.18)$, and using standard vector identities, with ${\bf r}_{\bf u}={\bf r}-{\bf u}r/b$, we have:    
\[
\begin{gathered}
  {\mathbf{E}}\left( {{\mathbf{x}},\tau } \right) =  - \frac{1}{b}\frac{{\partial {\mathbf{A}}\left( {{\mathbf{x}},\tau } \right)}}{{\partial \tau }} - \nabla \Phi \left( {{\mathbf{x}},\tau } \right) \hfill \\
   = \frac{{q{{\mathbf{r}}_u}\left( {1 - {{{{{{\mathbf{u}}^2}} \mathord{\left/
 {\vphantom {{{{\mathbf{u}}^2}} b}} \right.
 \kern-\nulldelimiterspace} b}}^2}} \right)}}{{{s^3}}} + \frac{{q\left[ {{\mathbf{r}} \times \left( {{{\mathbf{r}}_u} \times {\mathbf{a}}} \right)} \right]}}{{{b^2}{s^3}}} + \frac{{q\left( {{\mathbf{u}} \cdot {\mathbf{a}}} \right)\left[ {{\mathbf{r}} \times \left( {{\mathbf{u}} \times {\mathbf{r}}} \right)} \right]}}{{{b^4}{s^3}}}. \hfill \\ 
\end{gathered} 
\]
The computation of ${\bf B}$ is similar:  
\[
{\mathbf{B}} = \frac{{q\left( {{{\mathbf{r}}_u} \times {\mathbf{r}}} \right)\left( {1 - {{{{{{\mathbf{u}}^2}} \mathord{\left/
 {\vphantom {{{{\mathbf{u}}^2}} b}} \right.
 \kern-\nulldelimiterspace} b}}^2}} \right)}}{{{rs^3}}} + \frac{{q{\mathbf{r}} \times \left[ {{\mathbf{r}} \times \left( {{{\mathbf{r}}_u} \times {\mathbf{a}}} \right)} \right]}}{{{rb^2}{s^3}}} + \frac{{q\left( {{\mathbf{u}} \cdot {\mathbf{a}}} \right)\left[ {{\mathbf{r}} \times {\mathbf{u}}} \right]}}{{{b^4}{s^3}}}.
\]

It is easy to see that ${\bf B}$ is orthogonal to ${\bf E}$.  The first two terms in the above two equations are the same as (19-13) and (19-14) in \cite{34} (pg. 299).  The last term in each case arises because of the dissipative terms in equations $(6.3)$ and $(6.7)$.  These  terms are zero if $\bf a$ is zero or orthogonal to $\bf u$.  In the first case, there is no radiation and the particle moves with constant velocity so that the field is massless. The second case depends on  the creation of motion which keeps $\bf a$ orthogonal to $\bf u$ (for example a betatron).   Since  ${\bf r}\times \left( {{\bf u}\times {\bf r}}
\right)=r^2{\bf u}-\left( {{\bf u}\cdot {\bf r}} \right){\bf r}$, we see that there is a
component along the direction of propagation (longitudinal).  (Thus,  the $\bf E$ field has a  longitudinal part.) This confirms our claim that the new dissipative term is equivalent to an effective mass.   This means that the cause for radiation reaction comes directly from the use of the local clock to formulate Maxwell's equations.  Thus, there is no need to assume advanced potentials, self-interaction or mass renormalization along with the Lorentz-Dirac equation in order to account for radiation reaction as is required when the observer clock is used (Dirac theory). Furthermore, no assumptions about the structure of the source are required (i.e., Poincar${\acute e}$ stresses).  
\begin{rem}
We conjecture that this effective mass is the actual source of the photoelectric effect and that the photon is a real particle of non-zero (dynamical) mass, which travels with the fields but is not a field in the normal sense.  If this conjecture is correct, radiation from a betatron (of any frequency)  exposed to a metal surface will not produce photo electrons.  Such an experiment is within reach with current equipment. There are other implications of this observation, but further reflection is required.
\end{rem}
\subsection{Radiated Energy} 
The difference in the calculated fields for the two representations, makes it important to also compute the radiated energy for the (local) dual theory and compare it with the standard formulation.  The radiated energy is determined by the Poynting vector, which is defined by ${\bf \mcP}=\left( {{c
\mathord{\left/ {\vphantom {c {4\pi }}} \right. \kern- \nulldelimiterspace} {4\pi }}}
\right)\left( {{\bf E}\times {\bf B}} \right)$.  We closely follow the calculations in \cite{28}.

To compute the angular distribution of the radiated energy, we must carefully  note
that the rate of radiation is the amount of energy lost by the charge in a time interval
$d\tau'$ during the emission of a signal $\left( {{{-dU} \mathord{\left/ {\vphantom
{{-dU} {d\tau' }}} \right. \kern-\nulldelimiterspace} {d\tau' }}} \right)$.  At a field point, the Poynting vector ${\bf \mcP}$ represents the energy flow per unit time measured at the present time ($\tau$).   With this understanding, the same approach
that leads to the above formula gives ${\bf \mcP}=\left( {{{\bar b} \mathord{\left/
{\vphantom {{\bar b} {4\pi }}} \right. \kern-\nulldelimiterspace} {4\pi }}}
\right)\left( {{\bf E}\times {\bf B}} \right)$ in the proper-time formulation.  We thus
obtain the rate of energy loss of a charged particle into a given infinitesimal solid
angle $d\Omega $ as 
\beqn{\lb{pv}} 
-{{dU} \over {d\tau' }}(\Omega )d\Omega =\left( {{{\bar b}
\mathord{\left/ {\vphantom {{\bar b} {4\pi }}} \right. \kern-\nulldelimiterspace}
{4\pi }}} \right)\left[ {{\bf n}\cdot \left( {{\bf E}\times {\bf B}} \right)} \right]{\bf
r}^2{{d\tau } \over {d\tau'  }}d\Omega .
\eeqn  
Using equation ({\rf{rtd}}), we get that $\left( {{{d\tau } \mathord{\left/ {\vphantom
{{d\tau } {d\tau' }}} \right. \kern-\nulldelimiterspace} {d\tau' }}} \right)={{bs}
\mathord{\left/ {\vphantom {{bs} {\bar br}}} \right. \kern-\nulldelimiterspace} {\bar br}}$, so that ({\rf{pv}}) becomes 
\beqn{\lb{pv1}} 
-{{dU} \over {d\tau' }}(\Omega )d\Omega =\left( {{b \mathord{\left/
{\vphantom {b {4\pi }}} \right. \kern- \nulldelimiterspace} {4\pi }}} \right)\left[
{{\bf n}\cdot \left( {{\bf E}\times {\bf B}} \right)} \right]rsd\Omega .
\eeqn
As is well-known, only those terms that fall off as $\left( {{1 \mathord{\left/
{\vphantom {1 r}} \right. \kern- \nulldelimiterspace} r}} \right)$ (the radiation terms)
 contribute to the integral of ({\rf{pv1}}).  It is easy to see that our theory gives the following radiation terms: 
\beqn{\lb{pv2}}
 {\bf E}_{rad}={{q\left\{ {{\bf r}\times
\left[ {{\bf r}_{\bf u}\times {\bf a}} \right]} \right\}} \over {b^2s^3}}+{{q\left( {{\bf
u}\cdot {\bf a}} \right)\left[ {{\bf r}\times \left( {{\bf u}\times {\bf r}} \right)}
\right]} \over {b^4s^3}}={\bf E}_{rad}^c+ {\bf E}_{rad}^d,
\eeqn 

\beqn{\lb{pv3}}
{\bf
B}_{rad}={{q{\bf r}\times \left\{ {{\bf r}\times \left[ {{\bf r}_{\bf u}\times {\bf a}}
\right]} \right\}} \over {rb^2s^3}}+{{qr\left( {{\bf u}\cdot {\bf a}} \right)\left( {{\bf
r}\times {\bf u}} \right)} \over {b^4s^3}}={\bf B}_{rad}^c+{\bf
B}_{rad}^d,
\eeqn 

where ${\bf E}_{rad}^c,{\bf B}_{rad}^c$ are of the
same form as the classical terms with $c$ replaced by $b$, ${\bf w}'$ by ${\bf u}$,
and $\dot {\bf w}'$ by ${\bf a}$.  The two terms ${\bf E}_{rad}^d,{\bf B}_{rad}^d$,
are new and come directly from the dissipation term in the wave equations. (Note the
characteristic ${{\left( {{\bf u}\cdot {\bf a}} \right)}\mathord{\left/ {\vphantom
{{\left( {{\bf u}\cdot {\bf a}} \right)} {b^4}}} \right. \kern-\nulldelimiterspace}
{b^4}}$.)  We can easily integrate the classical terms to see that 
\beqn{\lb{pv4}}
\iint_\Omega  {\left( { - d{U^c}/d\tau } \right)}d\Omega  = \frac{b}{{4\pi }}\iint_\Omega  {\left[ {{\mathbf{n}} \cdot \left( {{\mathbf{E}}_{rad}^c \times {\mathbf{B}}_{rad}^c} \right)} \right]}rsd\Omega  = \frac{2}{3}\frac{{{q^2}{{\left| {\mathbf{a}} \right|}^2}}}{{{b^3}}}
\eeqn
  
This agrees with the standard result for small proper- velocity and proper-acceleration
of the charge when $b\approx c$ and ${\bf a}\approx {{d{\bf w}} \mathord{\left/
{\vphantom {{d{\bf w}} {dt}}} \right. \kern- \nulldelimiterspace} {dt}}$.

In the general case, our theory gives additional effects because of the dissipative
terms.  To compute the integral of ({\rf{pv1}}), we use spherical coordinates with the
proper-velocity ${\bf u}$ directed along the positive z-axis.  Without loss of
generality, we orient the coordinate system so that the proper-acceleration ${\bf a}$
lies in the xz-plane.  Let ${\alpha}$ denote the acute angle between ${\bf a}$ and
${\bf u}$, and substitute ({\rf{pv2}}) and ({\rf{pv3}}) in ({\rf{pv1}}) to obtain 
\beqn{\lb{pv4}}
\begin{gathered}
   - \frac{{dU}}{{d\tau }}d\Omega  = \frac{{{q^2}{{\left| {\mathbf{a}} \right|}^2}}}{{4\pi {b^3}}}\left\{ {{{\left( {1 - {\beta ^2}\cos \theta } \right)}^{ - 4}}\left[ {1 - {{\sin }^2}\theta {{\sin }^2}\alpha \cos \phi  - {{\cos }^2}\theta {{\cos }^2}\alpha  - \tfrac{1}{2}\sin 2\theta \sin 2\alpha \cos \phi } \right]} \right. \hfill \\
   - 2\beta {\left( {1 - \beta \cos \theta } \right)^{ - 5}}\left( {{{\sin }^2}\theta \cos \alpha  - \tfrac{1}{2}\sin 2\theta \sin \alpha \cos \phi } \right)\chi  + {\beta ^2}{\sin ^2}\theta {\left( {1 - \beta \cos \theta } \right)^{ - 6}}\left. {{\chi ^2}} \right\} \hfill \\ 
\end{gathered} 
\eeqn
where 
\beqn{\lb{pv5}}
 \chi ={{b^2} \over {r\left| {\bf a}
\right|}}\left( {1- \beta ^2} \right)+\beta \cos \alpha \left( {1-{1 \over \beta}\cos \theta
} \right)-\sin \theta \sin \alpha \cos \phi ,
\eeqn   
and $\beta =\left({{{\left| {\bf u} \right|} \mathord{\left/ {\vphantom {{\left| {\bf u}
\right|} b}} \right. \kern-\nulldelimiterspace} b}} \right)$.

The integration of ({\rf{pv4}}) over the surface of the sphere is elementary, and we
obtain, after some extensive but easy computations (see the
appendix of \cite{27}): 
\beqn{\lb{pv6}}
\begin{gathered}
  \mathop {\lim }\limits_{r \to \infty } \iint_{{\Omega _r}} { - \frac{{dU}}{{d\tau }}d\Omega } \hfill \\
   = \frac{{2{q^2}{{\left| {\mathbf{a}} \right|}^2}}}{{3{b^3}}}{\left( {1 - {\beta ^2}} \right)^{ - 3}}\left[ {1 - \tfrac{1}{5}{\beta ^2}\left( {4 + {\beta ^2}} \right) + \tfrac{1}{5}{\beta ^2}\left( {6 + {\beta ^2}} \right){{\sin }^2}\alpha } \right]. \hfill \\ 
\end{gathered} 
\eeqn   
As can be seen, this result agrees with ({\rf{pv1}}) at the lowest order.  For comparison,
the same calculation using the observer's clock for the case of general orientation of
velocity ${{d{\bf x}'}\mathord{\left/ {\vphantom {{d{\bf x}'} {dt'}}} \right.
\kern-\nulldelimiterspace} {dt'}}$ and acceleration ${{d{\bf w}'}\mathord{\left/
{\vphantom {{d{\bf w}'} {dt'}}} \right. \kern-\nulldelimiterspace} {dt'}}$ is 
\beqn{\lb{pv7}}
\mathop {\lim }\limits_{r \to \infty } \int_{{\Omega _r}} { - \frac{{dU}}{{d\tau }}d\Omega }  = \frac{{2{q^2}{{\left| {{\mathbf{\dot w'}}} \right|}^2}}}{{3{c^3}}}{\left( {1 - {\beta ^2}} \right)^{ - 3}}\left[ {1 - {\beta ^2}{{\sin }^2}\alpha } \right], 
\eeqn 
 
where $\beta =\left({{{\left| {\bf w}' \right|} \mathord{\left/ {\vphantom {{\left| {\bf
w}' \right|} c}} \right.\kern-\nulldelimiterspace} c}} \right)$.

We observe that, in general, for an arbitrary angle ${\alpha}$ with ${0\le{\alpha}
\le{{\pi}/2}}$ and arbitrary ${\beta}$ between $0$ and $1$, our result does not agree
with ({\rf{pv6}}) even if we replace $b$ with $c$ and ${\bf a}$ with $d{\bf w}'/dt'$. These
relatively large changes may prove important in the study of the physical and quantum electronics of nano systems.
\subsection{\bf Proper-time Group}
In part I, we constructed the Poincar${\acute e}$ algebra for the global system and produced the transformation between scale factors.  This was sufficient to show that observers could share information when they knew their relative velocity.  In this section, we directly identify the new  (proper-time) transformation group at the particle level necessary to preserve the first postulate. As will be seen, this transformation is both nonlinear and nonlocal because $b$ is not constant in this case, but depends on $\tau$.  In this section, ${\bf{x}}={\bf{x}}_i$, ${\bf{x}}'={\bf{x}}'_i$, $\tau={\tau}_i$ and ${\bf{V}}$ is the relative velocity between two observers.

The standard (Lorentz) time transformations between two inertial observers can be written as
\beqn\lb{t1}
t' = \gamma ({\mathbf{V}})\left[ {t - {{{\mathbf{x}} \cdot {\mathbf{V}}} \mathord{\left/ {\vphantom {{{\mathbf{x}} \cdot {\mathbf{V}}} {c^2 }}} \right.
 \kern-\nulldelimiterspace} {c^2 }}} \right], \quad \quad \quad \quad {\text{      }}t = \gamma ({\mathbf{V}})\left[ {t' + {{{\mathbf{x'}} \cdot {\mathbf{V}}} \mathord{\left/
 {\vphantom {{{\mathbf{x'}} \cdot {\mathbf{V}}} {c^2 }}} \right. \kern-\nulldelimiterspace} {c^2 }}} \right].
\eeqn
We want to replace $t,\;t'$ by $\tau$.  To do this, use the relationship between $dt$ and $d\tau$ to get:
\beqn\lb{t2}
t = \tfrac{1}{c}\int_0^\tau  {b(s)} ds = \tfrac{1}{c}\bar b\tau ,\quad  t' = \tfrac{1}{c}\int_0^\tau  {b'(s)} ds = \tfrac{1}{c}\bar b'\tau, 
\eeqn
where we have used the mean value theorem of calculus to obtain the final result, so that both $\bar b$ and $\bar b'$ represent an earlier $\tau$-value of $b$ and $b'$ respectively.  Thus, the transformations represent explicit nonlinear and nonlocal relationships between $t,\;t'$ and $\tau$ (during interaction).     If we set
\[
{\mathbf{d}}^ *   = {{\mathbf{d}} \mathord{\left/ {\vphantom {{\mathbf{d}} {\gamma ({\mathbf{V}})}}} \right.\kern-\nulldelimiterspace} {\gamma ({\mathbf{V}})}} - (1 - \gamma ({\mathbf{V}}))\left[ {({{{\mathbf{V}} \cdot {\mathbf{d}})} \mathord{\left/
{\vphantom {{{\mathbf{V}} \cdot {\mathbf{d}})} {(\gamma ({\mathbf{V}}){\mathbf{V}}^2 }}} \right. \kern-\nulldelimiterspace} {(\gamma ({\mathbf{V}}){\mathbf{V}}^2 }})} \right]{\mathbf{V}},
\] 
we can write the transformations that fix $\tau$ as:
\beqn
\begin{gathered}
\quad {\mathbf{x'}} = \gamma ({\mathbf{V}})\left[ {{\mathbf{x}}^ *   - ({{\mathbf{V}} \mathord{\left/{\vphantom {{\mathbf{V}} c}} \right.\kern-\nulldelimiterspace} c})\bar b\tau } \right],\quad \quad \quad \quad \,{\mathbf{x}} = \gamma ({\mathbf{V}})\left[ {{\mathbf{x'}}^*   + ({{\mathbf{V}} \mathord{\left/{\vphantom {{\mathbf{V}} c}} \right.\kern-\nulldelimiterspace} c})\bar b'\tau } \right], \hfill \\
\quad\quad {\mathbf{u'}} = \gamma ({\mathbf{V}})\left[ {{\mathbf{u}}^ *   - ({{\mathbf{V}} \mathord{\left/{\vphantom {{\mathbf{V}} c}} \right.\kern-\nulldelimiterspace} c})b} \right],\quad \quad \quad \quad \; \,{\text{   }}{\mathbf{u}} = \gamma ({\mathbf{V}})\left[ {{\mathbf{u'}}^*   + ({{\mathbf{V}} \mathord{\left/ {\vphantom {{\mathbf{V}} c}} \right. \kern-\nulldelimiterspace} c})b'} \right], \hfill \\
\quad {\mathbf{a'}} = \gamma ({\mathbf{V}})\left\{ {{\mathbf{a}}^*  - {\mathbf{V}}\left[ {({{{\mathbf{u}} \cdot {\mathbf{a}})} \mathord{\left/{\vphantom {{{\mathbf{u}} \cdot {\mathbf{a}})} {(bc}}} \right.\kern-\nulldelimiterspace} {(bc}})} \right]} \right\},\quad {\text{   }}{\mathbf{a}} = \gamma ({\mathbf{V}})\left\{ {{\mathbf{a'}}^*   + {\mathbf{V}}\left[ {({{{\mathbf{u'}} \cdot {\mathbf{a'}})} \mathord{\left/{\vphantom {{{\mathbf{u'}} \cdot {\mathbf{a'}})} {(b'c}}} \right. \kern-\nulldelimiterspace} {(b'c}})} \right]} \right\}. \hfill \\ 
\end{gathered} 
\eeqn
If we put equation (\rf{t2}) in (\rf{t1}), differentiate with respect to $\tau$  and cancel the extra factor of $c$, we get the transformations between $b$ and $b'$:
\beqn\lb{t3}
\quad \quad b'(\tau ) = \gamma ({\mathbf{V}})\left[ {b(\tau ) - {{{\mathbf{u}} \cdot {\mathbf{V}}} \mathord{\left/{\vphantom {{{\mathbf{u}} \cdot {\mathbf{V}}} c}} \right. \kern-\nulldelimiterspace} c}} \right],\quad \quad \quad \quad {\text{ }}b(\tau ) = \gamma ({\mathbf{V}})\left[ {b'(\tau ) + {{{\mathbf{u'}} \cdot {\mathbf{V}}} \mathord{\left/
 {\vphantom {{{\mathbf{u'}} \cdot {\mathbf{V}}} c}} \right. \kern-\nulldelimiterspace} c}} \right].
\eeqn
A version of  proper time group has been independently discussed in the works of A. A. Unger (see \cite{52, 53, 54}).

\subsubsection{The Transformation of Maxwell's equations}
It is shown in \cite{28}, that Maxwell's equations transform the same as in the conventional theory.  However, the current  and charge densities transform in the following manner:  
\beqn
 {\bf J}'={\bf J}+(\gamma -1){{({\bf J}\cdot {\bf V})} \over {{\bf
V}^2}}V-\gamma {b \over c}\rho {\bf V},  
\eeqn   
\beqn\lb{d1}
b'\rho' =\gamma \left[ {b\rho -({{{\bf J}\cdot {\bf V}} \mathord{\left/ {\vphantom {{{\bf J}\cdot
{\bf V}} c}} \right.\kern-\nulldelimiterspace} c})} \right]. 
\eeqn  
Using the first equation of (\rf{t3}) in (\rf{d1}), we have:  
\beqn\lb{d2}
 \rho' ={{\rho -({{{\bf
J}\cdot {\bf V}}\mathord{\left/ {\vphantom {{{\bf J}\cdot {\bf V}} {bc}}}\right.
\kern-\nulldelimiterspace} {bc}})}\over {1-({{{\bf u}\cdot {\bf V}} \mathord{\left/
{\vphantom {{{\bf u}\cdot {\bf V}} {bc}}} \right. \kern- \nulldelimiterspace}
{bc}})}}.  
\eeqn   
This differs from the standard result, which we obtain if we set $b'=b=c$ in (\rf{d1}):  
\[ 
\rho' =\gamma \left[ {\rho
-({{{\bf J}\cdot {\bf V}} \mathord{\left/ {\vphantom {{{\bf J}\cdot {\bf V}} {c^2}}}
\right.\kern-\nulldelimiterspace} {c^2}})} \right]. 
\]    
If we insert the expression ${{\bf J}/c}={\bf \rho} ({{\bf
u}/b})$ in (\rf{d2}); 
we obtain:  
\beqn\lb{d3}
\rho' =\rho {{1-({{{\bf u}\cdot {\bf V}} \mathord{\left/ {\vphantom {{{\bf u}\cdot
{\bf V}} {b^2}}} \right. \kern-\nulldelimiterspace} {b^2}})} \over {1-({{{\bf u}\cdot
{\bf V}} \mathord{\left/ {\vphantom {{{\bf u}\cdot {\bf V}} {bc}}} \right. \kern-
\nulldelimiterspace} {bc}})}}. 
\eeqn 
To see the impact of equation (\rf{d3}), suppose that a (arbitrary) charge distribution is at rest in the unprimed frame. From (\rf{d3}), we see that ${\bf u}=0$, so that $\rho'=\rho$.  Since the primed frame is arbitrary, the charge distribution will appear the same to all observers.  This is what we would expect on physical grounds, so that relatively moving frames should not change the symmetry properties of charged objects.  In particular, a charge distribution in one frame should not display physical effects due to another observer's relative motion.   
\subsection{Global Internal Dynamics}
In this section, we study the motion of one particle as seen from the global internal point of view. We assume that, if there is an external force on the system as a whole,  the  system as a whole has reached equilibrium. In this case, we have 
\[
{H_i} = {H_{i0}} + {V_i} = \sqrt {{c^2}{\bs{\pi}} _i^{_2} + m_i^{_2}{c^4}}  + {V_i},
\]
where ${\bs{\pi}}_i = {{\mathbf{p}}_i} - \tfrac{{{e_i}}}{c}{{\mathbf{A}}_i},\;{{\mathbf{A}}_i} = \sum\nolimits_{j \ne i} {{{\mathbf{A}}_{ji}}} $ and ${V_i} = \sum\nolimits_{j \ne i} {{V_{ji}}} $.
We assume that ${{\mathbf{A}}_{ji}}, {{V}_{ji}}$ represents the action of the retarded vector potential (respectively scalar potential) of the $j$-th particle on the $i$-th particle. Since, in this case, ${e_i}{{\mathbf{A}}_{ji}} \ne {e_j}{{\mathbf{A}}_{ij}}$ (respectively, ${V_{ji}} \ne {V_{ij}}$), we do not include the customary factor of $1/2$ in our definition of the scalar and vector potentials for particle $i$ (see \cite{27}).

Recall that $ {{\mathbf{w}}_i} = d{{\mathbf{x}}_i}/dt$ and $ {{\mathbf{u}}_i} = d{{\mathbf{x}}_i}/d{\tau _i}$. We define $ {{\mathbf{v}}_i} = d{{\mathbf{x}}_i}/d{\tau}$.  From our identities, its easy to show that
\[
\frac{{{{\mathbf{w}}_i}}}{c} = \frac{{{{\mathbf{v}}_i}}}{b} = \frac{{{{\mathbf{u}}_i}}}{{{b_i}}} \Rightarrow {\gamma _i^{-1}} = \sqrt {1 - {{\left( {\tfrac{{{{\mathbf{w}}_i}}}{c}} \right)}^2}}  = \sqrt {1 - {{\left( {\tfrac{{{{\mathbf{v}}_i}}}{b}} \right)}^2}}  = \sqrt {1 - {{\left( {\tfrac{{{{\mathbf{u}}_i}}}{{{b_i}}}} \right)}^2}} .
\]
The velocity ${{\mathbf{v}}_i}$ is the one our observer sees when he uses the global canonical
proper-clock ($\tau$), of the system to compute the particle velocity, while ${{\mathbf{w}}_i}$ is the
one seen when he uses his clock to compute the particle velocity.  if ${{\mathbf{U}}}$  is zero, $ b = c$ and, from the global perspective, our theory looks like the conventional one.  As the system is closed,  ${{\mathbf{U}}}$ is constant and $\tau$ is  linearly related to $t$.  Since $\gamma _i^{ - 1} = \tfrac{1}{b}\sqrt {{{\mathbf{U}}^2} + {c^2} - {\mathbf{v}}_i^2} $, the physical interpretation is very different if ${{\mathbf{U}}}$ is not zero. Furthermore, it is easy to see that, even if ${{\mathbf{U}}}$ is zero in one frame, it will not be zero in any other frame which is in relative motion.  Using $K$, the equations of motion are: 
\[
\begin{gathered}
  {{\mathbf{v}}_i} = \frac{{d{{\mathbf{x}}_i}}}{{d\tau }} = \frac{{\partial K}}{{\partial {{\mathbf{p}}_i}}} = \frac{H}{{M{c^2}}}\frac{{{c^2} {\bs{\pi}_i} }}{{{H_{i0}}}} = \frac{b}{c}\frac{{{c^2} {\bs{\pi}} _i}}{{{H_{i0}}}}, \hfill \\
 \frac{{d{{\mathbf{p}}_i}}}{{d\tau }} = \frac{{\partial K}}{{\partial {{\mathbf{p}}_i}}} = \frac{H}{{M{c^2}}}\sum\limits_{k = 1}^n {\left[ {\frac{{{c^2}{\bs{\nabla} _i}{{\bs{\pi}} _k}}}{{{H_{i0}}}} - {\bs{\nabla}  _i}{V_k}} \right]}  = \frac{b}{c}\sum\limits_{k = 1}^n {\left[ {\frac{{{c^2}{\bs{\nabla}  _i}{{\bs{\pi}} _k}}}{{{H_{i0}}}} - {\bs{\nabla}  _i}{V_k}} \right]} . \hfill \\ 
\end{gathered} 
\]
Factoring out the $k=i$ term $\tfrac{{{e_i}}}{b}\left[ {\left( {{{\mathbf{v}}_i} \cdot {\nabla _i}} \right){{\mathbf{A}}_i} + {{\mathbf{v}}_i} \times \left( {{\nabla _i} \times {{\mathbf{A}}_i}} \right)} \right]$, we have:
\beqn\lb{F1}
\begin{gathered}
  \frac{c}{b}\frac{{d{{\mathbf{p}}_i}}}{{d\tau }} = \tfrac{{{e_i}}}{b}\left[ {\left( {{{\mathbf{v}}_i} \cdot {\nabla _i}} \right){{\mathbf{A}}_i} + {{\mathbf{v}}_i} \times \left( {{\nabla _i} \times {{\mathbf{A}}_i}} \right)} \right] \hfill \\
   + \sum\limits_{k \ne i}^n {\left\{ {\frac{{{e_k}}}{b}\left[ {\left( {{{\mathbf{v}}_k} \cdot {\nabla _i}} \right){{\mathbf{A}}_k} + {{\mathbf{v}}_k} \times \left( {{\nabla _i} \times {{\mathbf{A}}_k}} \right)} \right] - {\nabla _i}{V_k}} \right\}} . \hfill \\ 
\end{gathered}
\eeqn
Using
\[
\left( {{{\mathbf{v}}_i} \cdot {\nabla _i}} \right){{\mathbf{A}}_i} = \frac{{d{{\mathbf{A}}_i}}}{{d\tau }} - \frac{{\partial {{\mathbf{A}}_i}}}{{\partial \tau }},
\]
equation (\rf{F1}) becomes
\beqn\lb{F2}
\begin{gathered}
  \frac{c}{b}\frac{{d{{\mathbf{p}}_i}}}{{d\tau }} - \tfrac{{{e_i}}}{b}\frac{{d{{\mathbf{A}}_i}}}{{d\tau }} = \tfrac{{{e_i}}}{b}\left[ {{{\mathbf{v}}_i} \times {{\mathbf{B}}_i}} \right] - \tfrac{{{e_i}}}{b}\frac{{\partial {{\mathbf{A}}_i}}}{{\partial \tau }} - {\nabla _i}{V_i} \hfill \\
   + \sum\limits_{k \ne i}^n {\left\{ {\frac{{{e_k}}}{b}\left[ {\left( {{{\mathbf{v}}_k} \cdot {\nabla _i}} \right){{\mathbf{A}}_k} + {{\mathbf{v}}_k} \times \left( {{\nabla _i} \times {{\mathbf{A}}_k}} \right)} \right] - {\nabla _i}{V_k}} \right\}} . \hfill \\ 
\end{gathered} 
\eeqn
Note that equation (\rf{F2}) can also be written as:
\beqn\lb{F3}
\begin{gathered}
  \frac{{d{{\mathbf{p}}_i}}}{{dt}} - \tfrac{{{e_i}}}{c}\frac{{d{{\mathbf{A}}_i}}}{{dt}} = \tfrac{{{e_i}}}{c}\left[ {{{\mathbf{w}}_i} \times {{\mathbf{B}}_i}} \right] - \tfrac{{{e_i}}}{c}\frac{{\partial {{\mathbf{A}}_i}}}{{\partial t }} - {\nabla _i}{V_i} \hfill \\
   + \sum\limits_{k \ne i}^n {\left\{ {\frac{{{e_k}}}{c}\left[ {\left( {{{\mathbf{w}}_k} \cdot {\nabla _i}} \right){{\mathbf{A}}_k} + {{\mathbf{w}}_k} \times \left( {{\nabla _i} \times {{\mathbf{A}}_k}} \right)} \right] - {\nabla _i}{V_k}} \right\}} . \hfill \\ 
\end{gathered}
\eeqn
Thus, equations (\rf{F2}) and (\rf{F3}) are mathematically equivalent.  Set ${V_i} = e_i{\Phi _i}$ and ${{\mathbf{E}}_i} =  - \tfrac{1}{b}\left( {\partial {{\mathbf{A}}_i}/\partial \tau } \right) - {\nabla _i}{\Phi _i}$, then we can write: 
\[
{{\mathbf{F}}_i} = \frac{{{e_i}}}{b}\left( {{{\mathbf{v}}_i} \times {{\mathbf{B}}_i}} \right) - \frac{{{e_i}}}{b}\frac{{\partial {{\mathbf{A}}_i}}}{{\partial \tau }} - {\nabla _i}{V_i} = {e_i}{{\mathbf{E}}_i} + \frac{{{e_i}}}{b}\left( {{{\mathbf{v}}_i} \times {{\mathbf{B}}_i}} \right).
\]
We can then write equation (\rf{F2}) as:
\[
\begin{gathered}
  \frac{c}{b}\frac{{d{\pi _i}}}{{d\tau }} = {{\mathbf{F}}_i} \hfill \\
   - \sum\limits_{k \ne i}^n {\left\{ {\frac{{{e_k}}}{b}\left[ {\left( {{{\mathbf{v}}_k} \cdot {\nabla _k}} \right){{\mathbf{A}}_{ik}} + {{\mathbf{v}}_k} \times \left( {{{\mathbf{v}}_k} \times {{\mathbf{A}}_{ik}}} \right)} \right] - {\nabla _k}{V_{ik}}} \right\}} . \hfill \\ 
\end{gathered} 
\]
If we now use
\[
\begin{gathered}
  \left( {{{\mathbf{v}}_k} \cdot {\nabla _k}} \right){{\mathbf{A}}_{ik}} = \frac{{d{{\mathbf{A}}_{ik}}}}{{d\tau }} - \frac{{\partial {{\mathbf{A}}_{ik}}}}{{\partial \tau }},\;\;{{\mathbf{B}}_{ik}} = {\nabla _k} \times {{\mathbf{A}}_{ik}}, \hfill \\
  {{\mathbf{E}}_{ik}} =  - \frac{1}{b}\frac{{\partial {{\mathbf{A}}_{ik}}}}{{\partial \tau }} - {\nabla _k}{\Phi _{ik}},\quad {{\mathbf{F}}_{ik}} = {e_k}{{\mathbf{E}}_{ik}} + \frac{{{e_k}}}{b}\left( {{{\mathbf{v}}_k} \times {{\mathbf{B}}_{ik}}} \right), \hfill \\ 
\end{gathered} 
\]
the above becomes:
\beqn{\lb{rr}}
\frac{c}{b}\frac{{d{\bs{\pi} _i}}}{{d\tau }} = {{\mathbf{F}}_i} - \sum\limits_{k \ne i}^n {\left[ {{{\mathbf{F}}_{ik}} - \frac{{{e_k}}}{b}\frac{{d{{\mathbf{A}}_{ik}}}}{{d\tau }}} \right]} .
\eeqn
If we simplify and put the last term on the other side, we have:
\[
\frac{c}{b}\sum\limits_{i = 1}^n {\frac{{d{\pi _i}}}{{d\tau }} + \sum\limits_{i = 1}^n {\sum\limits_{k \ne i}^n {\frac{{{e_k}}}{b}\frac{{d{{\mathbf{A}}_{ik}}}}{{d\tau }}} } }  = \sum\limits_{i = 1}^n {{{\mathbf{F}}_i}}  - \sum\limits_{i = 1}^n {\sum\limits_{k \ne i}^n {{{\mathbf{F}}_{ik}}} }. 
\]
Performing the summations on both sides give us:
\[
\begin{gathered}
  \frac{c}{b}\sum\limits_{i = 1}^n {\frac{{d{\pi _i}}}{{d\tau }} + } \sum\limits_{i = 1}^n {\frac{{{e_k}}}{b}\frac{{d{{\mathbf{A}}_{ik}}}}{{d\tau }}}  = 0 \Rightarrow  \hfill \\
  \sum\limits_{i = 1}^n {\frac{{d{{\mathbf{p}}_i}}}{{d\tau }}}  = 0 = \frac{{d{\mathbf{P}}}}{{d\tau }}. \hfill \\ 
\end{gathered} 
\]
\subsection{Discussion}
We want to first discuss the relationship between equation ({\rf{ipd2}}) and equation ({\rf{rr}}).  For comparison, we first rewrite equation ({\rf{ipd2}}) with its indices:
\beqn{\lb{rr1}}
\frac{c}{{{b_i}}}\frac{{d{\bs{\pi} _i}}}{{d{\tau _i}}} = \left[ {{e_i}{{\mathbf{E}}_i} + \tfrac{{{e_i}}}{{{b_i}}}\left( {{{\mathbf{u}}_i} \times {{\mathbf{B}}_i}} \right)} \right] = \frac{{d{\bs{\pi} _i}}}{{dt}}.
\eeqn
If we use
\[
{{\mathbf{F}}_i} = {e_i}{{\mathbf{E}}_i} + \frac{{{e_i}}}{b}\left( {{{\mathbf{v}}_i} \times {{\mathbf{B}}_i}}, \right).
\]
We can  write equation ({\rf{rr}}) as
\beqn{\lb{rr2}}
\frac{{d{\bs{\pi} _i}}}{{dt }} =\frac{c}{b}\frac{{d{\bs{\pi} _i}}}{{d\tau }} = \left[  {e_i}{{\mathbf{E}}_i} + \frac{{{e_i}}}{b}\left( {{{\mathbf{v}}_i} \times {{\mathbf{B}}_i}}, \right) \right] - \sum\limits_{k \ne i}^n {\left[ {{{\mathbf{F}}_{ik}} - \frac{{{e_k}}}{b}\frac{{d{{\mathbf{A}}_{ik}}}}{{d\tau }}} \right]} .
\eeqn
Equation ({\rf{rr1}}) represents one particle in a field of force, as seen locally.  It does not react via action at a distance, but its reaction shows up in its field via the additional term in its wave equation.  When we look at the same particle from the center of mass frame (equation {\rf{rr2}})), we see the  force which acts on the particle and the action at a distance reaction force of the particle on all the other particles in the system.  
 
We interpret the extra term on the (far) right-hand side of equation ({\rf{rr2}}) as the long-sought back reaction field of the $i$-th particle on all the other particles (the cause for its acceleration). Furthermore, this term accounts for radiation reaction without the Lorentz-Dirac equation, self-energy (divergence), advanced potentials or any assumptions about the structure of the source.  It is important to point out that the mathematical equivalence is manifest in both cases and yet these equations cannot be obtained if we start with the observer clock.

It also follows that equation ({\rf{rr2}}) is consistent with conservation of global momentum.  This along with   conservation of total energy implies the following:
\begin{cor}{\rm{(Wheeler-Feynman)}} In the $({\bf{X}}, t)$ or $({\bf{X}}, \tau)$ variables, the closed system of interacting charged particles exchange energy and momentum via fields and photons (action at a distance) and all emitted energy and momentum is absorbed internally. 
\end{cor}
Thus, the absorption hypothesis of  Wheeler and Feynman is automatically satisfied, without the use of advanced solutions to Maxwell's equations.
\subsubsection{Relationship to Quantum Theory}
Suppose there are only two particles interacting.  The retarded nature of the potentials means that their interactions will always be slightly off target, so that the law of action and reaction is only approximate at this level.  The more particles involved, the more likely that the reaction radiation will miss its intended mark.  The assumption that the system is closed and Theorem 4.3 implies that, after a short time, the system will reach equilibrium with particles, fields and  photons.  Thus, blackbody radiation is a natural consequence of charged particle interaction in a closed system.  This also suggests that the quantum behavior we observe in atoms is a consequence  of the same mechanism. 
\section{Part III: Implications and Applications}
\subsection{Newtonian Clock}
A fundamental conclusion of the last two sections is that, for any system of particles there always exists a unique observer-independent measure of time.  One consequence is the following theorem, which we first proved in \cite{28}.
\begin{thm}{\lb{Newtonian}}
Suppose that the observable universe is homogenous, isotropic and representable in the sense that it is independent of our observed portion of the universe.  Then the universe has a unique clock that is available to all observers.
\end{thm}
\begin{proof} Under the stated conditions ${H \mathord{\left/
 {\vphantom {H {M{c^2}}}} \right.
 \kern-\nulldelimiterspace} {M{c^2}}}$ is constant for our observed portion of the universe.  Since this property is observer independent, every observer will obtain the same ratio.  Thus, for any two observer's
\[  
 d{\tau} = \frac{{M{c^2}}}{H}dt=\frac{{M{c^2}}}{H}dt', \quad \Rightarrow t=t' \triangleq \tau_N.  
 \]
It follows that $\tau_N$ is uniquely defined.
\end{proof}   
\begin{thm} {\tx{(Peebles)}}
Suppose all observers choose a frame that is at rest with respect to the 2.7 $^{\circ}$K microwave background radiation, then  all laws will be invariant and not just covariant with respect to Lorentz transformations.
\end{thm}
In the study of physical systems one may not be interested in the behavior of the global system, but  in some subsystem.  The cluster decomposition property is a requirement of any theory purporting to represent  the real world.  This is the property that, if any two or more subsystems become widely separated, then they may can treated as independent systems (clusters).  We prove the following  in \cite{28}.
\begin{thm}
Suppose the system of particles can be decomposed into two or more clusters.  Then there exists a unique (local) clock and corresponding canonical Hamiltonian for each cluster. 
\end{thm}
\subsection{The Big Bang}
The current cosmological model for the universe assumes that it began around 13.8 billion years ago from a singularity (hot big bang).    There was no before, at one moment there was nothing and at the next moment the singularity appeared.  The theory only proposes to explain the time after this event when the singularity begin to expand to the universe we see today.   It is not a complete theory in the normal sense of a physical model.  The model leaves the following questions unanswered:
\begin{enumerate}
\item How is it that the universe appears to be so close to flat and uniform on a scale of almost 10 billion light years (flatness problem)?
\item How is it that regions in causally disconnected parts of space and time appear to have the same physical properties (horizon problem)?
\item How is it that we see matter and have not detected equal amounts of antimatter? 
\item How does the universe begin without conservation of energy (second law of thermodynamics)?
\item How does the universe begin without conservation of linear and angular momentum?
\end{enumerate}
The theory of (cosmic ) inflation  was introduced  by A. Guth in 1981 \cite{29} to provide a solution to the first two problems. In our view, the last three questions are equally (if not more) important.  Inflation assumes that, immediately after the big bang,  a superluminal (exponential) expansion rate happened, so that the space between any two points expanded faster than (our assumed) speed of light could travel between them.  This expansion solved the flatness problem and the horizon problem. The observable universe inflated from a very small volume and quickly became flat homogeneous and isotropic. The  inflation hypothesis is a solution to the first two problems, but lacking any evidence for a field that drives it, the theory has many critics and other approaches have been suggested (see \cite{ 30, 31, 32, 33} and references).  

At a minimum, any model of the beginning should be consistent with our current (experimentally obtained) understanding about the laws of the universe:
\begin{itemize}
\item Whenever an antiparticle is observed in experiment, we always find that it is also accompanied by a particle. 
\item  Whenever an interaction is observed in experiment, a complete analysis always shows that energy is conserved.
\item  Whenever an interaction is observed in experiment, a complete analysis always shows that linear momentum and angular momentum are conserved.
\end{itemize}
There are no compelling reasons for these known laws of physics to be violated for the big bang model to be true.  In this section, we suggest a slight alteration of the beginning, which brings the big bang model in line with our experimental understanding of the universe without the need for inflation,  (unknown or observed) new particles, fields, dark energy or other hypothetical devices.

We first revisit our conceptual view of the real numbers and their representation. Recall that a field is a set $A$ that has two binary operations $\oplus$ and $\odot$ that satisfies all our common experience with real numbers.  Formally:
\begin{Def}The real numbers is a triplet $(\R, +, \cdot)$, which is a field, with $0$ as the additive identity (i.e., $a+0=a$ for all $a \in \R$) and $1$ as the multiplicative identity (i.e., $a\cdot 1=a$ for all $a \in \R$). 
\end{Def}
This structure was designed by mathematicians without regard to its possible use in physics.  As a consequence the structural asymmetry went unnoticed and physicists accepted it without investigation until Santilli \cite{34} defined the isodual number field.  His definition is more general. For our purpose, we only need  the following.
\begin{Def}The isodual real numbers $(\hat{\R}, +, *)$ is a field, with $0=\hat{0}$ as the additive identity (i.e., $\hat{a}+\hat{0}=\hat{a}$ for all $-a=\hat{a} \in \hat{\R}$) and $\hat{1}=-1$ as the multiplicative identity (i.e., $\hat a * \hat 1 = ( - a)( - 1)( - 1) = \hat a$ for all $\hat{a} \in \hat{\R}$). 
\end{Def}
We note that we can obtain the isodual of any physical quantity $\hat{A}$ from the  equation  $A+\hat{A}=0$.
\begin{ex}
A simple example from quantum theory is the following:  the evolution of particle is defined on a Hilbert space $\mcH$ over the complex numbers $\C=\R + i\R$, with Hamiltonian $H$ by the equation:
\[
i\hbar \frac{{\partial \psi }}{{\partial t}} = H\psi .
\]
The conjugate equation is:
\[ 
- i\hbar \frac{{\partial {\psi ^*}}}{{\partial t}} = H{\psi ^*}.
\]
If we use ${\hat{\C}}$ as our number field, we can write the above equation as:
\[
\hat{i}*\hat{\hbar} * \frac{{\partial {\psi ^*}}}{{\partial \hat t}} = {\hat{H}}*{\psi ^*}
\]
Thus, this approach allows us to naturally  view anti-particles as time reversed particles, with their evolution defined on $\mcH^*$ over ${\hat{\C}}$. 
\end{ex}
\begin{rem} Santilli \cite{34} has shown that charge conjugation and isoduallity are equivalent for  the particle-antiparticle symmetry operation.  However,  isoduallity allows us to view the existence of antimatter and charge conservation as fundamental aspects of the universe, while also explaining why large amounts of antimatter are not found in this universe.    
\end{rem}
In the diagram below, we provide a new picture of the big bang beginning.  In this case, two universes are created, one going forward in $\tau_N$ and one going backward in $\tau_N$ (Newtonian time), relative to our reality.

Our solution follows the suggestions of Moffet \cite{34}.  His varying speed of light hypothesis is  consistent with the use of $b = \sqrt {{U^2} + {c^2}}$ for the fine tuning mechanism, with $U$ sufficiently large.    After equilibrium is reached, $U$ can slow down to zero, while $b$ reduces to $c$.  This would explain the flatness and horizon problems, but requires no new particles, fields, dark energy or other devices.
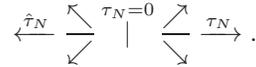
\[
\xleftarrow{{{{\hat \tau }_N}}}\frac{ \nwarrow }{ \swarrow }\mathop |\limits^{{\tau _N} = 0} \frac{ \nearrow }{ \searrow }\xrightarrow{{{\tau _N}}}.
\]
This view has the following advantages: we obtain
\begin{enumerate}
\item a natural arrow for time, with a zero initial point.
\item a natural explanation for the lack of antimatter in this universe.
\item antiparticles as particles moving backward in (Newtonian) proper time.
\item conservation of energy, linear and angular momentum.
\end{enumerate}
It is important to be clear that our assumption does not imply that there are any other symmetries or necessary similarities between the two universes. 

After writing this review, we came across a paper by Nielsen and Ninomiya \cite{43}, which also suggests a time reversed theory as an approach to saving the second law of thermodynamics.
\subsubsection{The Problem of Origin}
The possibility of other causes for the 2.7$^{\circ}$K mbr have been suggested in the past.
In a recent series of papers, Ares de Parga and co-workers \cite{44, 45, 46, 47, 48, 49, 50}  have developed a complete and consistent theory of relativistic thermodynamics.  They have extended it to both classical statistical mechanics and quantum statistical mechanics and obtained  the particle number density due to Henry \cite{51}.   As an application, they have studied the superposition of the radiation distributions from a number of  blackbodies radiating at different  temperatures and were able to reproduce the present 2.7$^{\circ}$K mbr.   This led them to suggest that the 2.7$^{\circ}$K mbr may have a  galactic origin.  In \cite{48} they proposed a feasible experiment that will  determine if such a cause is possible.
\section{Conclusion}
In this paper, we have discussed the only possible direct approaches to a relativistic theory for two or more particles.  These approaches are: that of Minkowski, that of Einstein and that of the recently discovered dual to the latter. We provide a table below, comparing the three approaches.    The one supported by the Minkowski postulate is the least complete of the three.   The Einstein and the dual version are both physically and mathematically consistent for any number of particles and have mathematically equivalent equations of motion.  However, the particle wave equations for their fields are not mathematically equivalent.  The dual version contains an additional longitudinal radiation term that appears instantaneously with  acceleration and does not depend on the nature of the force. This version predicts  photons are particles with nonzero effective dynamical mass.  It further predicts that radiation from a betatron of any frequency will not produce photoelectrons.  At the global level, the Wheeler-Feynman absorption hypothesis is a corollary of both the Einstein and dual theories, without  advanced potentials or any assumptions about the structure of the source.  This also proves the Wheeler-Feynman conjecture that action at a distance and field theory are complimentary aspects of each other.

By introducing a symmetric view of the number line, we have modified the standard version of the big bang to provide an arrow for time, explain the lack antimatter in the universe,  explain the flatness problem, the horizon problem, provide conservation of time, energy,  linear and angular momentum,  without inflation or any additional hypothesis. In addition, we predict that matter and antimatter are gravitationally repulsive.

The most important conclusion from this investigation is that the interpretation of experimental observations is not unique.  In order to make this last statement explicit in a very powerful manner, recall that many measurements are based on the dimensionless ratio $\beta  = {{\mathbf{w}} \mathord{\left/
 {\vphantom {{\mathbf{w}} c}} \right.     \kern-\nulldelimiterspace} c}$.  However, 
${{\mathbf{w}} \mathord{\left/   {\vphantom {{\mathbf{w}} c}} \right.
 \kern-\nulldelimiterspace} c} \equiv {{\mathbf{u}} \mathord{\left/
 {\vphantom {{\mathbf{u}} b}} \right.   \kern-\nulldelimiterspace} b}$, so we see that measurements of velocity and the speed of light for distant objects are totally ambiguous.  For example, (\cite{35}, pp. 556-561), the red shift factor $z$, used to determined distances in astronomy, is defined by:
\[
z = \sqrt {\frac{{1 + \tfrac{{\mathbf{w}}}
{c}}}
{{1 - \tfrac{{\mathbf{w}} }
{c}}}}  - 1 \equiv \sqrt {\frac{{1 + \tfrac{{\mathbf{u}}}
{b}}}
{{1 - \tfrac{{\mathbf{u}}}
{b}}}}  - 1.
\]
We thus conclude that distant objects may have much higher velocities and light may have a velocities higher than $c$, without any contradiction.    
\newpage

\begin{tabular}{r|c||c|c}
{\bf For $n>1$}: \quad\quad \quad	& {\bf Minkowski} & {\bf Einstein} & {\bf Dual} \\
\hline \hline
{\Small \bf{reference frame }} &  {\Small inertial } & {\Small inertial }  & {\Small inertial }\\
\hline
{\Small \bf speed of light } & {\Small independent of source} & {\Small independent of source} & {\Small depends on source} \\
\hline
{\Small \bf space-time} & {\Small dependent variables} & {\Small independent variables} & {\Small independent variables} \\
\hline
{\Small \bf transfomation group} & {\Small linear Lorentz} & {\Small linear Lorentz} & {\Small nonlinear Lorentz} \\
\hline
{\Small \bf cluster property} & {\Small  non-existent } & {\Small possible theory} & {\Small general theory} \\
\hline
{\Small \bf many-particle} & {\Small non-existent} & {\Small possible theory} & {\Small general theory}\\
\hline
{\Small \bf radiation reaction} & {\Small highly problematic} & {\Small partial theory} & {\Small complete  theory} \\
\hline
{\Small \bf quantum theory } & {\Small non-existent } & {\Small follows from theory}  & {\Small follows from theory} \\
\hline
{\Small \bf arrow for time } & {\Small non-existent} & {\Small follows from  theory}  & {\Small follows from theory} \\
\hline
{\Small \bf universal clock} & {\Small non-existent} & {\Small follows from  theory} & {\Small follows from theory} \\
\hline
{\Small \bf big bang} & {\Small possible  theory} & {\Small possible  theory} & {\Small possible theory} \\
\hline
{\Small \bf theory of gravity } & {\Small possible  theory*} & {\Small possible  theory**} & {\Small possible theory**} \\
\hline
\end{tabular}
\noindent{\SMALL *The general theory of relativity.} 
{\SMALL **The program suggested by Dirac (see quotes before section 1.2).}\\  
\section{\bf Acknowledments}  

We would like to thank Professors J. Dunning-Davis,  L. Horwitz, M. C. Land, E. Leib and R. M. Santilli    
 for their continued interest and many constructive suggestions.

\bibliographystyle{amsalpha}

\end{document}